\documentclass[aps,pra,10pt,twocolumn,superscriptaddress,showkeys]{revtex4-2}
\usepackage{amsfonts}
\usepackage{amsmath}
\usepackage{mathtools}
\usepackage{algorithm}
\usepackage{algorithmic}
\usepackage{xfrac}
\usepackage{amssymb}
\usepackage{graphicx}
\usepackage[colorlinks=true,linkcolor=blue,citecolor=red,plainpages=false,pdfpagelabels]
{hyperref}
\usepackage{xcolor}
\usepackage{qcircuit}
\providecommand{\U}[1]{\protect\rule{.1in}{.1in}}

\newtheorem{theorem}{Theorem}

\newtheorem{lemma}[theorem]{Lemma}

\newenvironment{proof}[1][Proof]{\noindent\textbf{#1.} }{\ \rule{0.5em}{0.5em}}

\allowdisplaybreaks

\begin{document}
\preprint{ }

\title[ ]{Quantum thermodynamics and semi-definite optimization}

\author{Nana Liu}
\affiliation{Institute of Natural Sciences, Shanghai Jiao Tong University, Shanghai 200240, China}
\affiliation{School of Mathematical Sciences, Shanghai Jiao Tong University, Shanghai
200240, China}
\author{Michele Minervini}
\affiliation{School of Electrical and Computer Engineering, Cornell University, Ithaca, New
York 14850, United States}
\author{Dhrumil Patel}
\affiliation{School of Computer Science, Cornell University, Ithaca, New York 14850, United States}
\author{Mark M. Wilde}
\affiliation{School of Electrical and Computer Engineering, Cornell University, Ithaca, New
York 14850, United States}


\begin{abstract}
In quantum thermodynamics, a system is described by a Hamiltonian and a list
of non-commuting charges representing conserved quantities like particle
number or electric charge, and an important goal is to determine the system's
minimum energy in the presence of these conserved charges. In optimization
theory, a semi-definite program (SDP) involves a linear objective function optimized
over the cone of positive semi-definite operators intersected with an affine
space. These problems arise from differing motivations in the physics and
optimization communities and are phrased using very different terminology, yet
they are essentially identical mathematically. By adopting Jaynes' mindset
motivated by quantum thermodynamics, we observe that minimizing free energy in
the aforementioned thermodynamics problem, instead of energy, leads to an
elegant solution in terms of a dual chemical potential maximization problem that is concave in
the chemical potential parameters. As such,
one can employ standard (stochastic) gradient ascent methods to find the
optimal values of these parameters, and these methods are guaranteed to
converge quickly. At low temperature, the minimum free energy provides an
excellent approximation for the minimum energy. We then show how this Jaynes-inspired gradient-ascent approach can be used in both first- and second-order classical and
hybrid quantum--classical algorithms for minimizing energy, and equivalently, how it can be used
for solving SDPs, with guarantees on the runtimes of the
algorithms. The approach discussed here is well grounded in quantum
thermodynamics and, as such, provides physical motivation underpinning why
algorithms published fifty years after Jaynes' seminal work, including the
matrix multiplicative weights update method, the matrix exponentiated gradient
update method, and their quantum algorithmic generalizations, perform well at solving SDPs.
\end{abstract}
\date{\today}
\startpage{1}
\endpage{10}
\maketitle
\tableofcontents

\section{Introduction}

\subsection{Background and motivation}

Quantum mechanics and semi-definite optimization are inevitably intertwined.
The postulates of quantum mechanics assert that a state of a quantum system is
specified by a positive semi-definite operator with trace equal to one. They
also assert that a measurement is specified by a tuple of positive
semi-definite operators that sum to the identity. Beyond states and
measurements, quantum channels are general models of physical
evolutions~\cite{Kraus1983}, and these are in one-to-one correspondence with
their Choi matrices~\cite{Choi1975}, which are positive semi-definite
operators that reduce to the identity matrix after a partial trace over the
channel output system. All of these matrix constraints are semi-definite
constraints, and so all of the core constituents of quantum mechanics are
represented as constrained semi-definite operators.

As such, when one is faced with certain optimization tasks in quantum
mechanics, like computing the minimum energy of a physical system or the
optimal error probabilities for communication~\cite{Yuen1970,Yuen1975}, one
soon thereafter realizes that they can be formulated as semi-definite
optimization problems, also known as semi-definite programs (SDPs). Beyond
these basic problems, various problems in many-body
physics~\cite{Mazziotti2004,Barthel2012,SimmonsDuffin2015,Berenstein2023,Fawzi2024}
and quantum communication~\cite{W18thesis} can also be addressed by SDPs, so
that the field of semi-definite optimization has greatly contributed to
quantum mechanics. Spurred by the large interest in quantum information
science, SDPs have become an indispensable tool in the
field~\cite{khatri2020principles,Siddhu2022,Skrzypczyk2023}.

Despite the aforementioned intimate connection between semi-definite
optimization and quantum mechanics, the former domain has its roots in
operations research and computer science, and in particular, a special class of
semi-definite optimization problems called linear optimization or linear
programming~\cite[Section~4.3]{Boyd2004}. Indeed, many problems of practical
interest in these and other fields can be formulated as or relaxed to SDPs,
including problems in combinatorial
optimization~\cite{Goemans1997,Rendl1999,tunccel2016polyhedral},
finance~\cite{Gepp2020}, job scheduling in operations
research~\cite{Skutella2001}, and machine
learning~\cite{dAspremont2007,hall2018optimization,majumdar2020recent}, to
name a few. These practical applications have motivated the development of
fast algorithms for solving SDPs, which are based on intuition coming from
earlier algorithms for solving linear programs. Some prominent algorithms
include interior-point methods~\cite{Alizadeh1998,Todd1998,Jiang2020}, the
matrix exponentiated gradient update method~\cite{Tsuda2005}, and the very closely-related
matrix multiplicative weights update method
\cite{Arora2007,Arora2012,Arora2016}.

In light of the above, one might be left to wonder:\ given all that
semi-definite optimization has contributed to quantum mechanics, can quantum
mechanics also offer insights into semi-definite optimization? Hitherto, this
question has been addressed in a particular way by the development of various
algorithms for solving SDPs on quantum computers, some with guaranteed
runtimes
\cite{Brandao2017,Apeldoorn2019,Brandao2019,vanApeldoorn2020quantumsdpsolvers,Kerenidis2020,Augustino2023quantuminterior,watts2023quantum}
and others based on heuristics
\cite{Bharti2022,Patel2024variationalquantum,Patti2023quantumgoemans,westerheim2023dualvqequantumalgorithmlower,chen2023qslackslackvariableapproachvariational}. However, another way of addressing this question is to employ physical
intuition and insights for understanding semi-definite optimization and developing
algorithms for solving SDPs. The latter is the main theme of the present paper.

\subsection{Summary of contributions}

\label{sec:summary-contribs}

In this paper, we show how physical intuition coming from quantum
thermodynamics can be used for understanding semi-definite optimization. In
particular, we consider the problem of minimizing a physical system's energy
in the presence of conserved non-commuting charges. This is a fundamental
problem in quantum thermodynamics and has its roots in Jaynes' seminal work on
unifying information theory and statistical mechanics
\cite{Jaynes1957,Jaynes1957a,Jaynes1962}. It has also resurfaced in recent works on quantum information and thermodynamics~\cite{YungerHalpern2016,Guryanova2016,Lostaglio2017,YungerHalpern2020,Anshu2021,Kranzl2023,Majidy2023}. By glancing at the precise
formulation of this problem in~\eqref{eq:energy-minimization}, it should be
clear to readers having background in semi-definite optimization that
\eqref{eq:energy-minimization} is nearly in the standard form of an SDP
\cite[Eq.~(4.51)]{Boyd2004}, with the only exception being that there is a
unit trace constraint imposed on the positive semi-definite operator $\rho$ therein in
order for it to be a legitimate quantum state.

\begin{figure}
\centering
\includegraphics[width=0.60\linewidth]{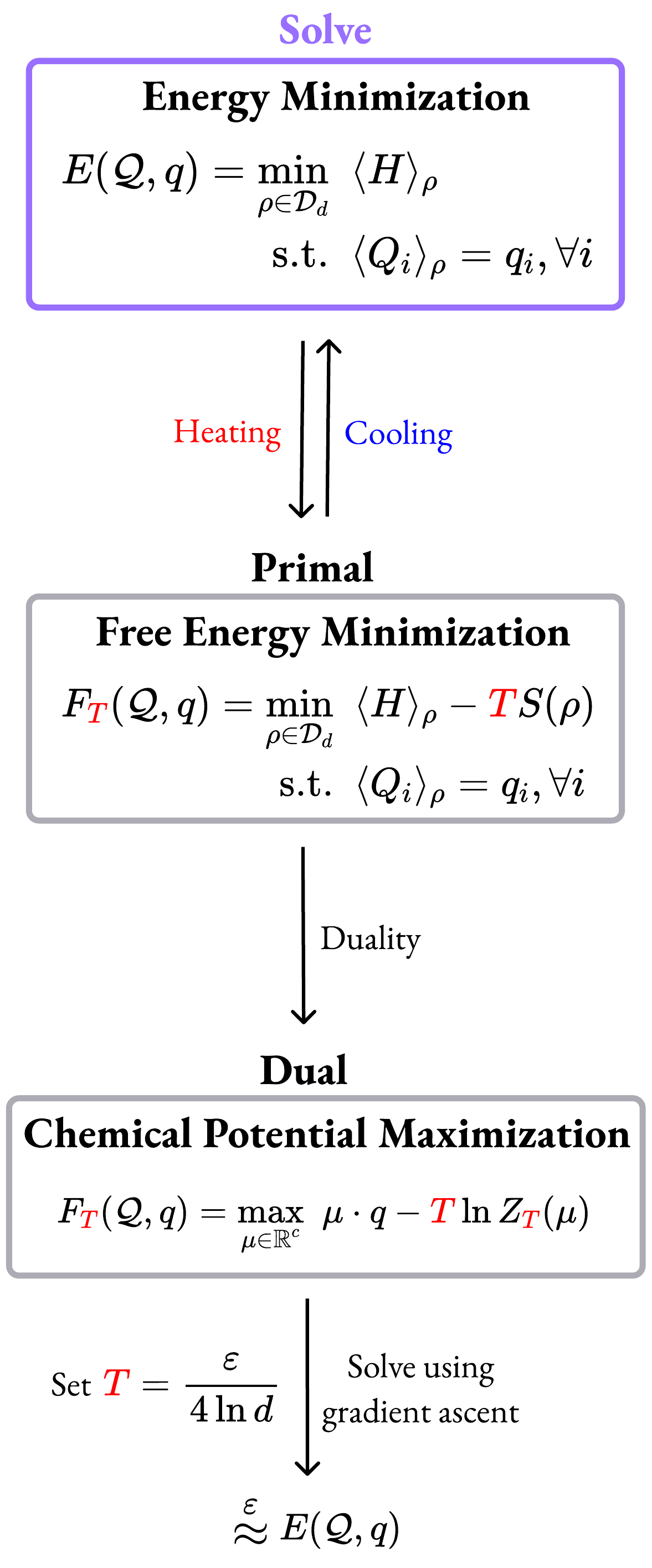}
\caption{Depiction of the main idea behind solving an energy minimization
problem in quantum thermodynamics, specified by a Hamiltonian $H$ and a tuple
$(Q_{1}, \ldots, Q_{c})$ of conserved non-commuting charges with respective
expected values $(q_{1}, \ldots, q_{c})$. The goal is to determine the minimum
energy $E(\mathcal{Q},q)$ of the system. Inspired by the fact that physical
systems operate at a strictly positive temperature $T>0$, we instead minimize the free
energy $F_{T}(\mathcal{Q},q)$ of the system at a low temperature $T$. By
employing Lagrangian duality and quantum relative entropy, we can rewrite the free energy minimization
problem as the dual problem of chemical potential maximization. This latter
problem is concave in the chemical potential vector $\mu$ and thus can be
solved quickly by gradient ascent. The approach leads to classical and hybrid
quantum--classical algorithms for energy minimization. See Section~\ref{sec:energy-min-q-thermo} for details.}
\label{fig:energy-minimization}
\end{figure}

Our physics-inspired approach for solving~\eqref{eq:energy-minimization}
consists of a few observations:

\begin{enumerate}
\item Rather than directly solve the energy minimization problem in
\eqref{eq:energy-minimization}, one can instead minimize the free energy of
the system, defined in~\eqref{eq:min-free-energy}, at a strictly positive
temperature, which gives an excellent approximation to the minimum energy at
sufficiently low temperature (in particular, ``sufficiently low'' means that the temperature should be  inversely proportional to the number of qubits in the system).

\item By employing Lagrangian duality and quantum relative entropy, one can
express the free-energy minimization problem in dual form as a chemical
potential maximization problem, i.e., a maximization of the Legendre
transformation of a certain log-partition function (see
\eqref{eq:dual-free-energy}). The maximization is taken over a vector $\mu$,
which is known as the chemical potential vector in thermodynamics.

\item The latter maximization problem is concave in the chemical potential
vector~$\mu$ and can thus be solved quickly by gradient ascent.
\end{enumerate}

\noindent This line of reasoning is depicted in
Figure~\ref{fig:energy-minimization}, and the precise gradient ascent
algorithm is given as Algorithm~\ref{alg:basic-grad-asc-free-energy} in Section~\ref{sec:solution-gradient-ascent}.
Interestingly, Jaynes already made the observation stated in step~2 above when he
considered the related problem of maximizing the entropy of a thermodynamic
system described by a Hamiltonian and non-commuting conserved charges~\cite[page~200]{Jaynes1962}. The first observation in step~3, regarding concavity, was
attributed to Jaynes before the statement of~\cite[Proposition~3]{liu2006gibbsstatesconsistencylocal}.

The analysis in step~2 demonstrates that parameterized thermal states of the
form in~\eqref{eq:grand-canon-thermal-state} are optimal for minimizing the
free energy. Thus, one can confine the search to states of this form, and the
concavity claim in step~3 simplifies the search for the optimal state even
more, such that the chemical potential maximization can be performed by
gradient ascent and is guaranteed to converge quickly. Such parameterized
thermal states go by alternative names in the literature, including grand
canonical thermal states, non-Abelian thermal states
\cite{YungerHalpern2016,YungerHalpern2020,Kranzl2023,Majidy2023}, and quantum
Boltzmann machines~\cite{Amin2018,Benedetti2017,Kieferova2017}.

Not only does Algorithm~\ref{alg:basic-grad-asc-free-energy} specify a
classical algorithm, but it also leads to a hybrid quantum--classical (HQC)
algorithm for performing energy minimization (see
Algorithm~\ref{alg:energy-min}). The only change needed in promoting
Algorithm~\ref{alg:basic-grad-asc-free-energy} to an HQC algorithm is that
expressions like $\left\langle Q_{i}\right\rangle _{\rho_{T}(\mu)}$ and
$\left\langle H\right\rangle _{\rho_{T}(\mu)}$ should instead be estimated on
a quantum computer, by preparing the thermal state $\rho_{T}(\mu)$ and
measuring the observables $Q_{i}$ and $H$. This HQC algorithm is thus broadly
similar to the HQC algorithm proposed in~\cite{Coopmans2024} for generative
modeling of quantum states, and as such, Algorithm~\ref{alg:energy-min} can be considered a
\textit{quantum Boltzmann machine learning} algorithm for energy minimization. As mentioned, it is necessary for the HQC algorithm to make use of thermal-state preparation algorithms as a subroutine, and we note that there has been much progress in this area in recent years~\cite{chen2023q_Gibbs_sampl,chen2023thermalstatepreparation,rajakumar2024gibbssampling,bergamaschi2024gibbs_sampling,chen2024sim_Lindblad,rouze2024efficientthermalization,bakshi2024hightemperaturegibbsstates,ding2024preparationlowtemperaturegibbs}.  Due to a
quantum computer only producing estimates of the values $\left\langle
Q_{i}\right\rangle _{\rho_{T}(\mu)}$ and $\left\langle H\right\rangle
_{\rho_{T}(\mu)}$, it is necessary to analyze the performance of
Algorithm~\ref{alg:energy-min} under the framework of stochastic gradient
ascent (see~\cite[Chapter~5]{garrigos2024handbookconvergencetheoremsstochastic} and~\cite[Section~6]{Bubeck2015}).
However, again due to concavity of the maximization problem, it is guaranteed
that, on average, stochastic gradient ascent converges quickly to
an approximately optimal solution.

\begin{figure}
\centering
\includegraphics[width=0.60\linewidth]{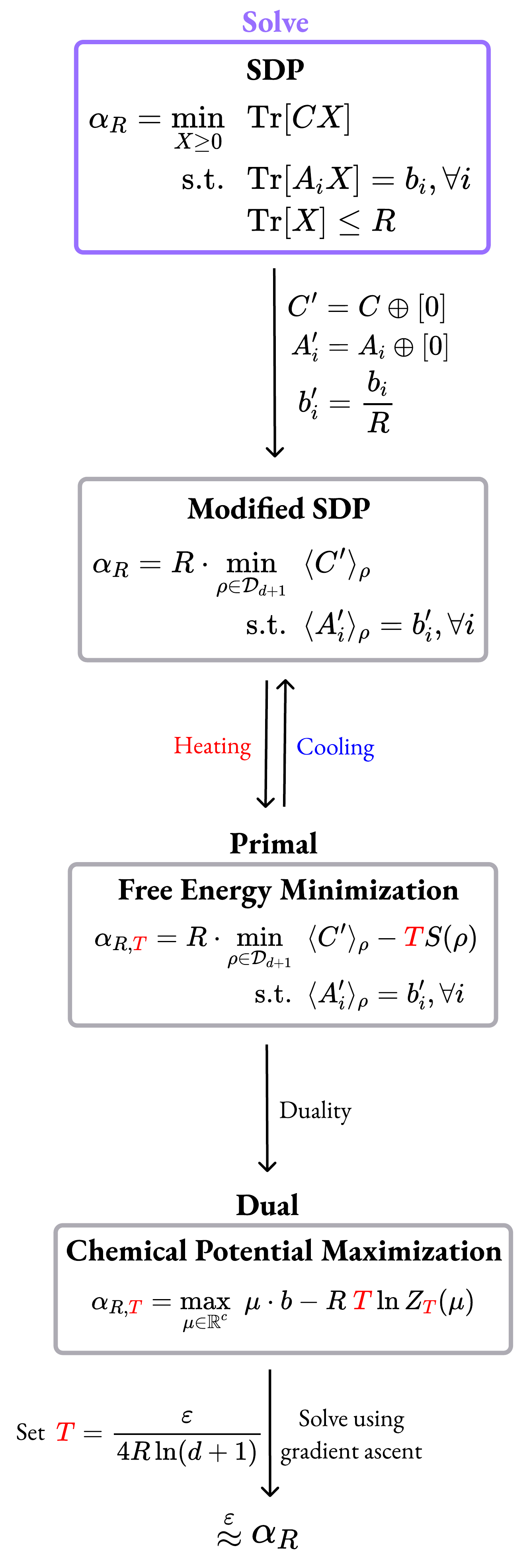}
\caption{Depiction of the main idea behind solving a general semi-definite
optimization problem (abbreviated SDP), specified by the tuple of $d\times d$ Hermitian
matrices $(C,A_{1}, \ldots, A_{c})$. The reduction from~\cite[Footnote~3]
{Brandao2019} allows for reducing a general SDP to an energy minimization
problem, where $R \geq0$ is a guess on the trace of an optimal solution to the
original SDP. The resulting algorithm is sample-efficient as long as $R$ does not
grow too quickly with $d$. See Section~\ref{sec:SDP-to-energy-min} for details.}
\label{fig:general-SDP}
\end{figure}

Going beyond energy minimization problems of the form in~\eqref{eq:energy-minimization}, one can modify this approach just slightly in
order to solve general SDPs of the standard form in~\eqref{eq:SDP-standard}.
This makes use of a simple reduction pointed out in~\cite[Footnote~3]{Brandao2019}.
Indeed, as recalled in Section~\ref{sec:general-SDPs} and summarized in
Figure~\ref{fig:general-SDP}, one can reduce a general SDP\ to a scaled
energy-minimization problem, such that the overall runtime ends up being a
function of the scaling factor. Furthermore, this leads to both classical and
HQC algorithms for solving general SDPs.

Algorithms~\ref{alg:basic-grad-asc-free-energy}\ and \ref{alg:energy-min}
are first-order methods for solving energy minimization problems (i.e., which
use the gradient in a gradient ascent update rule in order to search for an
optimal solution). By employing an analytical form for the Hessian of the objective function in~\eqref{eq:dual-free-energy}, as stated in~\cite[Lemma~26]{Anshu2021}, we establish a second-order method for energy
minimization and semi-definite optimization, i.e., which incorporates
second-derivative information to aid the search. This leads to second-order
classical and HQC algorithms for these tasks. The second-order HQC algorithms employ Hessian estimation algorithms similar to those recently proposed in~\cite{patel2024quantumboltzmannmachinelearning,patel2024naturalgradientparameterestimation,minervini2025evolvedquantumboltzmannmachines,minervini2025quantumnaturalgradientthermalstate}, involving a combination of random sampling, Hamiltonian simulation
\cite{lloyd1996universal,childs2018toward}, and the Hadamard test
\cite{Cleve1998}.
Interestingly, for our case, the Hessian matrix is equal to the negative of
the Kubo--Mori information matrix~\cite{Bengtsson2006,Liu2019,Sidhu2020,Jarzyna2020,Meyer2021fisherinformationin,sbahi2022provablyefficientvariationalgenerative}, indicating that the Euclidean geometry on the chemical potential parameter
space is aligned with the geometry induced by the quantum relative entropy on parameterized thermal states of the form in~\eqref{eq:grand-canon-thermal-state}. Related, this indicates that the
second-order method is equivalent to natural gradient ascent according to the Kubo--Mori
geometry induced by quantum relative entropy, the latter detailed recently in~\cite{patel2024naturalgradientparameterestimation,minervini2025evolvedquantumboltzmannmachines}.

The rest of our paper is organized as follows. In Section~\ref{sec:energy-min-q-thermo}, we begin by setting up the energy minimization problem in quantum thermodynamics, and we then show that it can be approximated by free energy minimization, which is in turn related to the dual problem of chemical potential maximization. This dual problem can be solved quickly by gradient ascent. See Figure~\ref{fig:energy-minimization} for a summary of the method. In Section~\ref{sec:algorithm-analysis}, we analyze the smoothness parameter for the chemical potential maximization problem, which is essential for establishing a lower bound on the number of steps required for gradient ascent to converge. In Section~\ref{sec:SDP-to-energy-min}, we recast a general semi-definite optimization problem as an energy minimization problem and remark how the latter can be used to solve the former. See Figure~\ref{fig:general-SDP} for a summary of the method. In Section~\ref{sec:q-algs-SDPs}, we propose HQC algorithms for energy minimization and semi-definite optimization, and we also analyze their convergence under stochastic gradient ascent. In Section~\ref{sec:second-order}, we discuss  a second-order approach to gradient ascent that involves HQC algorithms for estimating the elements of the Hessian matrix. In Section~\ref{sec:MEG-MMW-compare}, we compare the approach put forward here to other approaches like the matrix exponentiated gradient update method~\cite{Tsuda2005},  the matrix multiplicative weights update method~\cite{Arora2007,Arora2012,Arora2016}, and a quantum algorithm for solving SDPs~\cite{Apeldoorn2019}. Finally, in Section~\ref{sec:conclusion}, we conclude with a summary of our findings and open directions for future research. 

\section{Minimizing energy in quantum thermodynamics}

\label{sec:energy-min-q-thermo}

\subsection{Energy minimization problem}

\label{sec:energy-min-problem}

Minimizing energy and the related task of minimizing free energy are both problems of
fundamental interest in quantum thermodynamics
\cite{Jaynes1957,Jaynes1957a,Jaynes1962}, and related problems have resurfaced
in recent works on quantum information
\cite{YungerHalpern2016,Guryanova2016,Lostaglio2017,YungerHalpern2020,Anshu2021,Kranzl2023,Majidy2023}. To set up the problem, given is a Hamiltonian~$H$ and a tuple $\left(
Q_{1},\ldots,Q_{c}\right)  $ of non-commuting conserved charges corresponding
to a $d$-dimensional quantum system, where $c,d\in\mathbb{N}$. Define
\begin{equation}
\mathcal{Q}\equiv\left(  H,Q_{1},\ldots,Q_{c}\right)  .
\label{eq:observable-tuple}
\end{equation}
Mathematically, each of these matrices is a $d\times d$ Hermitian matrix; they
are also called observables in quantum mechanics. Physically, the Hamiltonian
$H$ describes the energy of the system, and the non-commuting charges $\left(
Q_{1},\ldots,Q_{c}\right)  $ represent conserved quantities such as particle
number or electric charge. All of the operators in $\mathcal{Q}$ need not
commute. A constraint vector
\begin{equation}
q\equiv\left(  q_{1},\ldots,q_{c}\right)  \in\mathbb{R}^{c}
\end{equation}
specifies the expectation values of the conserved charges with respect to a
quantum state $\rho\in\mathcal{D}_{d}$ of the system:
\begin{equation}
\left\langle Q_{i}\right\rangle _{\rho}=q_{i}\ \ \forall i\in\left[  c\right]
,
\end{equation}
where
\begin{equation}
\mathcal{D}_{d}\coloneqq\left\{  \rho\in\mathbb{M}_{d}:\rho\geq
0,\ \operatorname{Tr}[\rho]=1\right\}
\end{equation}
denotes the set of $d\times d$ density matrices (positive semi-definite
matrices with trace equal to one), $\mathbb{M}_{d}$ denotes the set of
$d\times d$ matrices with entries in $\mathbb{C}$, 
\begin{equation}
\left\langle O\right\rangle _{\rho}\equiv\operatorname{Tr}[O\rho]
\end{equation}
denotes the expected value of the observable $O$ in the state $\rho$, and
\begin{equation}
\left[  c\right]  \coloneqq\left\{  1,\ldots,c\right\}  .
\end{equation}
Minimizing the energy in the presence of the conserved charges  corresponds to the
following optimization task:
\begin{equation}
E(\mathcal{Q},q)\coloneqq\min_{\rho\in\mathcal{D}_{d}}\left\{  \left\langle
H\right\rangle _{\rho}:\left\langle Q_{i}\right\rangle _{\rho}=q_{i}\ \forall
i\in\left[  c\right]  \right\}  . \label{eq:energy-minimization}
\end{equation}

\subsection{Approximation by free energy minimization}

Finding an analytic form for the solution of~\eqref{eq:energy-minimization} or
the set of optimizing states is not obvious. However, one can consider that
real physical systems operate at a strictly positive temperature $T>0$ and instead seek
to determine the minimum free energy at temperature $T$:
\begin{equation}
F_{T}(\mathcal{Q},q)\coloneqq\min_{\rho\in\mathcal{D}_{d}}\left\{
\left\langle H\right\rangle _{\rho}-TS(\rho):\left\langle Q_{i}\right\rangle
_{\rho}=q_{i}\ \forall i\in\left[  c\right]  \right\}  ,
\label{eq:min-free-energy}
\end{equation}
where $\left\langle H\right\rangle _{\rho}-TS(\rho)$ is the free energy and
the von Neumann entropy is defined as
\begin{equation}
S(\rho)\coloneqq-\operatorname{Tr}[\rho\ln\rho].
\end{equation}

Determining the free energy $F_{T}(\mathcal{Q},q)$ when the temperature $T$ is
low gives an approximation for the energy $E(\mathcal{Q},q)$. To make the
approximation precise, one can employ the following bounds on the von Neumann
entropy of a $d$-dimensional state $\rho$:
\begin{equation}
0\leq S(\rho)\leq\ln d, \label{eq:von-neumann-entropy-ineqs}
\end{equation}
which lead to the bounds:
\begin{equation}
E(\mathcal{Q},q)\geq F_{T}(\mathcal{Q},q)\geq E(\mathcal{Q},q)-T\ln d.
\label{eq:energy-free-energy-bounds}
\end{equation}
With this approximation in place, Eq.~\eqref{eq:energy-free-energy-bounds}
clarifies that
\begin{equation}
\lim_{T\rightarrow0}F_{T}(\mathcal{Q},q)=E(\mathcal{Q},q)
\end{equation}
and that we should set $\delta=T\ln d$ in order to achieve a $\delta>0$ error when
approximating the energy $E(\mathcal{Q},q)$ in terms of the free energy
$F_{T}(\mathcal{Q},q)$. As such, we see that the temperature $T$ should be set
as low as $T=\sfrac{\delta}{\ln d}$ in order to achieve this error (i.e., the
temperature $T$ should be inversely proportional to the number of qubits in
the system). See Appendix~\ref{app:energy-free-energy-bnd} for further details
on establishing the bounds in~\eqref{eq:energy-free-energy-bounds}.

\subsection{Chemical potential maximization}

In order to determine the minimum energy $E(\mathcal{Q},q)$ up to error
$\delta$, we can thus instead determine the free energy $F_{T}(\mathcal{Q},q)$
at low temperature $T=\sfrac{\delta}{\ln d}$. By employing Lagrange
duality~\cite[Chapter~5]{Boyd2004} and the quantum relative
entropy~\cite{Umegaki1962} 
\begin{equation}
D(\omega\Vert\tau)\coloneqq\operatorname{Tr}[\omega\left(  \ln\omega-\ln
\tau\right)  ]
\end{equation}
of density matrices $\omega,\tau\in\mathcal{D}_{d}$ (also known as Umegaki relative entropy), one can mirror Jaynes' argument in~\cite[page~200]{Jaynes1962} to arrive at the
following alternative expression for the free energy~$F_{T}(\mathcal{Q}
,q)$:
\begin{equation}
F_{T}(\mathcal{Q},q)=\sup_{\mu\in\mathbb{R}^{c}}\left\{  \mu\cdot q-T\ln
Z_{T}(\mu)\right\}  , \label{eq:dual-free-energy}
\end{equation}
where $Z_{T}(\mu)$ denotes the following partition function:
\begin{equation}
Z_{T}(\mu)\coloneqq\operatorname{Tr}\!\left[  e^{-\frac{1}{T}\left(
H-\mu\cdot Q\right)  }\right]  ,
\end{equation}
and we adopt the shorthand
\begin{equation}
\mu\cdot Q\equiv\sum_{i\in\left[  c\right]  }\mu_{i}Q_{i}.
\end{equation}
See Appendix~\ref{app:proof-dual-expression} for a detailed proof of
\eqref{eq:dual-free-energy}. The vector $\mu\in\mathbb{R}^{c}$ corresponds
mathematically to Lagrange multipliers for the constraints in
\eqref{eq:min-free-energy}, and in quantum thermodynamics, its entries are
known as chemical potentials. By inspection, Eq.~\eqref{eq:dual-free-energy}
is the Legendre transformation or Fenchel dual of the temperature-scaled
log-partition function $T\ln Z_{T}(\mu)$.

Beyond the
equality in~\eqref{eq:dual-free-energy}, it also follows that the optimal
state $\rho$ in~\eqref{eq:min-free-energy} is unique and given by the grand
canonical thermal state $\rho_{T}(\mu^{\ast})$ of the following form
\cite[page~200]{Jaynes1962}:
\begin{equation}
\rho_{T}(\mu)\coloneqq\frac{e^{-\frac{1}{T}\left(  H-\mu\cdot Q\right)  }
}{Z_{T}(\mu)}, \label{eq:grand-canon-thermal-state}
\end{equation}
where $\mu^{\ast}$ denotes the optimal value in~\eqref{eq:dual-free-energy}.
The proof in Appendix~\ref{app:proof-dual-expression} clarifies this point. In
more recent literature, $\rho_{T}(\mu)$ is known as a non-Abelian thermal
state~\cite{YungerHalpern2016,YungerHalpern2020,Kranzl2023,Majidy2023} or a
quantum Boltzmann machine~\cite{Amin2018,Benedetti2017,Kieferova2017}.

\subsection{Solution using gradient ascent}

\label{sec:solution-gradient-ascent}

The optimal chemical potential vector $\mu^{\ast}$ depends on the fixed
constraint vector $q$ in~\eqref{eq:min-free-energy}, but it does not have a
closed form in general. However, the log-partition function $\ln Z_{T}(\mu)$
in~\eqref{eq:dual-free-energy} is a convex function of $\mu$ \footnote{The
proof of convexity is attributed to~\cite{Jaynes1962} just before
\cite[Proposition~3]{liu2006gibbsstatesconsistencylocal}.}, and Jaynes proved
that the elements of the gradient $\nabla_{\mu}\left(  \mu\cdot q-T\ln
Z_{T}(\mu)\right)  $ are given by~\cite[page~200]{Jaynes1962}
\begin{equation}
\frac{\partial}{\partial\mu_{i}}\left(  \mu\cdot q-T\ln Z_{T}(\mu)\right)
=q_{i}-\left\langle Q_{i}\right\rangle _{\rho_{T}(\mu)}.
\label{eq:gradient-formula}
\end{equation}
See Appendix~\ref{app:gradient-hessian} for a derivation
of~\eqref{eq:gradient-formula}. As such, the objective function $\mu\cdot q-T\ln Z_{T}(\mu)$ in \eqref{eq:dual-free-energy} is concave in $\mu$, and one can search for the value of $\mu$
that maximizes~\eqref{eq:dual-free-energy} by means of the standard gradient
ascent algorithm, which is guaranteed to converge to a point $\varepsilon
$-close to the global maximum of the objective function
in~\eqref{eq:dual-free-energy} in $O(1/\varepsilon)$ steps, where
$\varepsilon>0$~\cite[Corollary~3.5]
{garrigos2024handbookconvergencetheoremsstochastic}. See Figure~\ref{fig:energy-minimization} for a concise depiction of the method summarized in Sections~\ref{sec:energy-min-problem}--\ref{sec:solution-gradient-ascent}.

This completes the essential reasoning behind an algorithm for calculating the
minimum energy $E(\mathcal{Q},q)$, which we provide as Algorithm~\ref{alg:basic-grad-asc-free-energy} (the smoothness
parameter $L\geq0$ will be specified later in Section~\ref{sec:algorithm-analysis}):

\begin{algorithm}[H]
\caption{$\mathtt{minimize\_energy}(H, (Q_i)_i, (q_i)_i, L, d, \varepsilon, r)$}
\label{alg:basic-grad-asc-free-energy}
\begin{algorithmic}[1]
\STATE \textbf{Input:} \begin{itemize}\setlength\itemsep{-0.15em}
    \item Observables $H$ and  $(Q_i)_{i\in [c]}$
    \item Constraint values $(q_i)_{i\in[c]}$
    \item Smoothness parameter $L$
    \item Hilbert space dimension $d$
    \item Desired error $\varepsilon > 0$
    \item Radius $r$: An upper bound on $\|\mu^*\|$, where $\mu^*$ is the optimal solution to~\eqref{eq:dual-free-energy} for $T=\frac{\varepsilon}{4\ln d}$
\end{itemize}
\STATE Set $T \leftarrow \frac{\varepsilon}{4 \ln d}$
\STATE Initialize $\mu^0 \leftarrow (0, \ldots, 0)$
\STATE Set learning rate $\eta \in (0, \sfrac{1}{L}]$
\STATE Choose $M = \left\lceil\sfrac{L r^2}{\varepsilon}\right\rceil$
\FOR{$m = 1$ to $M$}
    \STATE $\mu^m \leftarrow \mu^{m-1} + \eta \left(q-\langle Q \rangle_{\rho_T(\mu^{m-1})} \right)$
\ENDFOR
\RETURN $\mu^M\cdot
q + \langle H  - \mu^M\cdot
 Q\rangle _{\rho_{T}(\mu^M)}$
\end{algorithmic}
\end{algorithm}

In Algorithm~\ref{alg:basic-grad-asc-free-energy}, we have employed the notation $q-\langle Q\rangle _{\rho_{T}(\mu^M)}$, which is a shorthand for the $c$-dimensional vector with its $i$th component given by $q_i-\langle Q_i\rangle _{\rho_{T}(\mu^M)}$, for all $i \in [c]$.

The equality in~\eqref{eq:dual-free-energy}, the inequality
in~\eqref{eq:energy-free-energy-bounds}, and~\cite[Corollary~3.5]{garrigos2024handbookconvergencetheoremsstochastic} guarantee that the output
\begin{align}
\tilde{E} & \equiv \langle H \rangle_{\rho_T(\mu^M)} + \mu^M\cdot
(q-\langle Q\rangle _{\rho_{T}(\mu^M)})    \\
& = \mu^M\cdot
q + \langle H  - \mu^M\cdot
 Q\rangle _{\rho_{T}(\mu^M)}
\end{align}
of
Algorithm~\ref{alg:basic-grad-asc-free-energy} is an $\varepsilon$-approximate
estimate of $E(\mathcal{Q},q)$. There are three sources of error in
Algorithm~\ref{alg:basic-grad-asc-free-energy}:\ the error from approximating
the minimum energy $E(\mathcal{Q},q)$ by the minimum free energy
$F_{T}(\mathcal{Q},q)$ at temperature~$T$, the error from gradient
ascent arriving at an approximate global minimum of the free energy, and the
error from outputting $\tilde{E}$
instead of
\begin{align}
&   \langle H \rangle_{\rho_T(\mu^M)} -TS(\rho_{T}
(\mu^{M})) + \mu^M\cdot
(q-\langle Q\rangle _{\rho_{T}(\mu^M)})  \nonumber \\
& =  \mu^M\cdot q-T\ln
Z_{T}(\mu^M). \label{eq:f-mu-to-free-energy-exp}
\end{align}
See \eqref{eq:log-part-back-to-free-energy-1}--\eqref{eq:log-part-back-to-free-energy-2} for a proof of \eqref{eq:f-mu-to-free-energy-exp}. Outputting $\tilde{E}$ is
simpler, as this expression involves just matrix traces, and it is more
amenable to generalizing Algorithm~\ref{alg:basic-grad-asc-free-energy} to an HQC algorithm, as done later on in Section~\ref{sec:q-algs-SDPs}. See
Appendix~\ref{app:alg-error-analysis}\ for a detailed error analysis of
Algorithm~\ref{alg:basic-grad-asc-free-energy}.

The starting point for Algorithm~\ref{alg:basic-grad-asc-free-energy} is
$\mu^{0}=\left(  0,\ldots,0\right)  $, which corresponds to setting the
initial state to be $\rho_{T}(0)=\frac{e^{-H/T}}{\operatorname{Tr}[e^{-H/T}]}
$, i.e., a thermal state of the Hamiltonian $H$ at temperature $T$. Intuitively, Algorithm~\ref{alg:basic-grad-asc-free-energy} proceeds by
increasing the magnitude of the $i$th chemical potential $\mu_{i}$ if the
$i$th constraint is violated as follows: $q_{i}>\left\langle Q_{i}\right\rangle $. If instead $q_{i}<\left\langle Q_{i}\right\rangle $, the algorithm proceeds to decrease the chemical potential $\mu_i$. When $q_{i}>\left\langle Q_{i}\right\rangle $ and we increase the chemical potential $\mu_i$, this has the thermodynamical interpretation of making the system more thermodynamically favorable to gaining charges from the reservoir and thus they are added to the system. The system responds by increasing $\left\langle Q_{i}\right\rangle$, which thus guides $\left\langle Q_{i}\right\rangle$ to be closer to $q_i$. Similarly, if $q_{i}<\left\langle Q_{i}\right\rangle $ and we decrease the chemical potential $\mu_i$ in the gradient ascent step, the system responds by discouraging more charge addition, hence decreasing $\left\langle Q_{i}\right\rangle$ to be closer to $q_i$. This process continues until there is a balance between the charges entering the system and those exiting, and charge conservation of the system, with expected charge being $q_i$, is reached. The
final state is then guaranteed to satisfy the constraints approximately while
minimizing the energy $\left\langle H\right\rangle $ approximately, with the
quality of the approximation guaranteed by the choice of the temperature $T$.

The main missing piece needed for understanding the convergence of
Algorithm~\ref{alg:basic-grad-asc-free-energy} is the smoothness parameter
$L$, which we address in Section~\ref{sec:algorithm-analysis}. Additionally,
one should have an upper bound $r$ on $\left\Vert \mu^{\ast}\right\Vert $ (i.e.,
amounting to a rough guess of $\mu^{\ast}$ or the size of the search space) in order to know the number of
steps, $M$, that are required for
Algorithm~\ref{alg:basic-grad-asc-free-energy} to converge. See
\eqref{eq:number-of-steps-bound-smoothness-1}--\eqref{eq:number-of-steps-bound-smoothness}
in Section~\ref{sec:algorithm-analysis} for a precise lower bound on the
number of steps, $M$, required by
Algorithm~\ref{alg:basic-grad-asc-free-energy} to reach an $\varepsilon$-approximate
estimate of~$E(\mathcal{Q},q)$.

Before moving on, let us remark here that
Algorithm~\ref{alg:basic-grad-asc-free-energy} serves as the basic blueprint
for an HQC algorithm for minimizing the free energy, as well as for
classical and HQC algorithms for solving general SDPs.

\section{Analysis of the smoothness parameter in
Algorithm~\ref{alg:basic-grad-asc-free-energy}}

\label{sec:algorithm-analysis}

As stated at the end of Section~\ref{sec:solution-gradient-ascent}, the main missing piece for establishing convergence of
Algorithm~\ref{alg:basic-grad-asc-free-energy} is the smoothness parameter
$L\geq0$. This parameter is defined to be a Lipschitz constant for the
gradient, i.e., a constant $L$ such that the following inequality
holds~\cite[Definition~2.24]{garrigos2024handbookconvergencetheoremsstochastic}:
\begin{equation}
\left\Vert \nabla_{\mu}f(\nu_{1})-\nabla_{\mu}f(\nu_{2})\right\Vert \leq
L\left\Vert \nu_{1}-\nu_{2}\right\Vert \quad\forall\nu_{1},\nu_{2}
\in\mathbb{R}^c,
\end{equation}
where, for our problem,
\begin{equation}
\label{eq:obj_fun}
f(\mu)\coloneqq\mu\cdot q-T\ln Z_{T}(\mu).
\end{equation}
As such, the smoothness parameter $L$ controls the maximum amount that the
gradient $\nabla_{\mu}f(\mu)$ can change when changing $\mu$ during the
gradient ascent search in Algorithm~\ref{alg:basic-grad-asc-free-energy}, and
thus it plays a role in determining how large of a stepsize one should take
when performing gradient ascent. Equivalently, for a smooth function, the
smoothness parameter $L$ can be taken as an upper bound on the largest
singular value of the Hessian matrix of $f(\mu)$, the latter defined in terms
of its matrix elements as $\frac{\partial^{2}}{\partial\mu_{i}\partial\mu_{j}
}f(\mu)$ \cite[Lemma~2.26]{garrigos2024handbookconvergencetheoremsstochastic}. As such, one can then resort to bounding the largest singular value
of the Hessian in order to determine the performance of gradient ascent.

Interestingly, Jaynes also identified a formula for the matrix elements of the
Hessian of $f(\mu)$~\cite[page~200]{Jaynes1962}, and the negative of this formula is now
known as the Kubo--Mori information matrix of the parameterized family
$\left(  \rho_{T}(\mu)\right)  _{\mu\in\mathbb{R}^{c}}$, a concept from
quantum information geometry
\cite{Bengtsson2006,Liu2019,Sidhu2020,Jarzyna2020,Meyer2021fisherinformationin,sbahi2022provablyefficientvariationalgenerative}. Specifically, Jaynes established the following formula for all
$i,j\in\left[  c\right]  $:
\begin{align}
\frac{\partial^{2}}{\partial\mu_{i}\partial\mu_{j}}f(\mu)  &  =-I_{ij}
^{\operatorname{KM}}(\mu),\label{eq:hess_-KM}\\
I_{ij}^{\operatorname{KM}}(\mu)  &  =\frac{1}{T}\int_{0}^{1}
ds\ \operatorname{Tr}[\rho_{T}(\mu)^{s}Q_{i}\rho_{T}(\mu)^{1-s}Q_{j}
]\nonumber\\
&  \qquad-\frac{1}{T}\left\langle Q_{i}\right\rangle _{\rho_{T}(\mu
)}\left\langle Q_{j}\right\rangle _{\rho_{T}(\mu)}.\label{eq:hessian_via_KM}
\end{align}
By applying the triangle inequality and a multivariate generalization of
H\"{o}lder's inequality, the following upper bound holds on the magnitude of
the matrix elements of the Hessian for all $i,j\in\left[  c\right]  $:
\begin{equation}
\left\vert \frac{\partial^{2}}{\partial\mu_{i}\partial\mu_{j}}f(\mu
)\right\vert \leq\frac{2}{T}\left\Vert Q_{i}\right\Vert \left\Vert
Q_{j}\right\Vert . \label{eq:hessian-elements-bound}
\end{equation}
The value $\frac{2}{T}
\sum_{i \in[c]}\left\Vert Q_{i}\right\Vert ^{2}$ serves as an upper bound  on the largest singular value
of the Hessian, because the largest singular value is bounded from above by the trace norm of the Hessian and the trace norm is bounded from above by $\frac{2}{T}
\sum_{i \in[c]}\left\Vert Q_{i}\right\Vert ^{2}$ in this case, which implies
that we can set the smoothness parameter $L$ as follows:
\begin{equation}
L=\frac{2}{T}
\sum_{i \in[c]}\left\Vert Q_{i}\right\Vert ^{2}
. \label{eq:smoothness-parameter-choice}
\end{equation}
See Lemmas~\ref{lem:upper-bnd-matrix-elements-hessian} and \ref{lem:hessian-upper-bound} in Appendix~\ref{app:gradient-hessian} for details
of~\eqref{eq:hessian-elements-bound} and \eqref{eq:smoothness-parameter-choice}, respectively.

Now revisiting Algorithm~\ref{alg:basic-grad-asc-free-energy}, it follows from~\eqref{eq:smoothness-parameter-choice} and the error
analysis in Appendix~\ref{app:alg-error-analysis} that Algorithm~\ref{alg:basic-grad-asc-free-energy}
converges to an $\varepsilon$-approximate optimal solution in a number of steps, $M$, which
satisfies
\begin{align}
M    = \left\lceil \frac{Lr^2 }{\varepsilon}\right\rceil 
\label{eq:number-of-steps-bound-smoothness-1}
&  =\left\lceil\left(\frac{2}{T}
\sum_{i \in[c]}\left\Vert Q_{i}\right\Vert ^{2}\right) \frac{r^2
}{\varepsilon} \right\rceil \\
&  = \left\lceil\frac{8 \, r^2 \ln d}{\varepsilon^2} 
\sum_{i \in[c]} \left\Vert Q_{i}\right\Vert ^{2}\right\rceil .
\label{eq:number-of-steps-bound-smoothness}
\end{align}
Furthermore, we can set the step size $\eta$ to satisfy
\begin{align}
    0 < \eta \leq \frac{1}{L} 
     & = \frac{T}{2
\sum_{i \in[c]}\left\Vert Q_{i}\right\Vert ^{2}} \\
& = \frac{\varepsilon}{8 (\ln d) \sum_{i \in[c]}\left\Vert Q_{i}\right\Vert ^{2}}.
\end{align}
This concludes our analysis of Algorithm~\ref{alg:basic-grad-asc-free-energy}.

\section{Semi-definite optimization as energy minimization}

\label{sec:SDP-to-energy-min}

\label{sec:general-SDPs}A standard form for a semi-definite optimization
problem, also called a semi-definite program (SDP),\ is as
follows~\cite[Eq.~(4.51)]{Boyd2004}:
\begin{equation}
\alpha\coloneqq\min_{X\geq0}\left\{  \operatorname{Tr}[CX]:\operatorname{Tr}
[A_{i}X]=b_{i}\ \forall i\in\left[  c\right]  \right\}
,\label{eq:SDP-standard}
\end{equation}
where $C$ is a $d\times d$ Hermitian matrix, so is each $A_{i}$, and $b_{i}
\in\mathbb{R}$ for all $i\in\left[  c\right]  $ (see also
\cite{wolkowicz2012handbook}).

Clearly, if one of the constraints in~\eqref{eq:SDP-standard} is
$\operatorname{Tr}[X]=1$ (i.e., corresponding to $A_{i}=I$ and $b_{i}=1$ for
some $i\in\left[  c\right]  $), then the SDP in
\eqref{eq:SDP-standard} is identical to the energy minimization problem in
\eqref{eq:energy-minimization}, with the identifications
\begin{equation}
C\leftrightarrow H, \quad A_{i}\leftrightarrow Q_{i}, \quad b_{i}
\leftrightarrow q_{i} \quad\forall i\in\left[  c\right]  .
\label{eq:identify-SDP-energy-min}
\end{equation}
In this case, one can then solve the SDP\ in~\eqref{eq:SDP-standard} up to
$\varepsilon$~error by means of Algorithm~\ref{alg:basic-grad-asc-free-energy}, using the same number of steps given in~\eqref{eq:number-of-steps-bound-smoothness-1}--\eqref{eq:number-of-steps-bound-smoothness}.

In the case that the constraint $\operatorname{Tr}[X]=1$ is not present, one
can make a few adjustments to the SDP\ in~\eqref{eq:SDP-standard} to reduce it
to an energy minimization problem of the form
in~\eqref{eq:energy-minimization}, as previously detailed in~\cite[Footnote~3]
{Brandao2019}. To begin with, let us define the following modified SDP:
\begin{equation}
\alpha_{R}\coloneqq\min_{X\geq0}\left\{  \operatorname{Tr}
[CX]:\operatorname{Tr}[X]\leq R,\ \operatorname{Tr}[A_{i}X]=b_{i}\ \forall
i\in\left[  c\right]  \right\}  , \label{eq:R-SDP}
\end{equation}
where $R>0$. Due to the fact that this is an extra constraint, the following
inequality holds
\begin{equation}
\alpha_{R}\geq\alpha
\end{equation}
and it is saturated in the limit as $R\rightarrow\infty$:
\begin{equation}
\lim_{R\rightarrow\infty}\alpha_{R}=\alpha.
\end{equation}
Here, $R$ represents a guess on the trace of an optimal solution
to~\eqref{eq:SDP-standard} and is similar in spirit to the guess $r$ on the norm
$\left\|  \mu^{*} \right\|  $ of the optimal dual variable $\mu^{*}$, which is
needed to determine the precise number of steps required for gradient ascent
(see discussion after Algorithm~\ref{alg:basic-grad-asc-free-energy}).

The following lemma establishes a reduction of the SDP\ optimization task in
\eqref{eq:R-SDP} to the energy minimization problem in~\eqref{eq:energy-minimization}.

\begin{lemma}
[\cite{Brandao2019}]\label{lem:reduction-SDP-to-energy-min} The following
equality holds:
\begin{equation}
\alpha_{R}=R\cdot\min_{\rho\in\mathcal{D}_{d+1}}\left\{  \left\langle
C^{\prime}\right\rangle _{\rho}:\left\langle A_{i}^{\prime}\right\rangle
_{\rho}=b_{i}^{\prime}\ \forall i\in\left[  c\right]  \right\}  ,
\label{eq:reduction-to-states}
\end{equation}
where
\begin{equation}
C^{\prime} \coloneqq C\oplus\left[  0\right]  ,\ A_{i}^{\prime}
\coloneqq A_{i}\oplus\left[  0\right]  ,\ b_{i}^{\prime} \coloneqq \frac
{b_{i}}{R} \quad\forall i\in\left[  c\right]  .
\end{equation}

\end{lemma}

\begin{proof}
See Appendix~\ref{app:reduction-SDP-to-energy-min}.
\end{proof}

\medskip

Given the reduction from Lemma~\ref{lem:reduction-SDP-to-energy-min}, it is
thus possible to solve SDPs of the general form in~\eqref{eq:SDP-standard} by
means of Algorithm~\ref{alg:basic-grad-asc-free-energy}.  Under the reduction
from Lemma~\ref{lem:reduction-SDP-to-energy-min}, however, it is necessary to
solve the energy minimization problem in~\eqref{eq:reduction-to-states} to
error $\sfrac{\varepsilon}{R}$ in order for the overall error in approximating
$\alpha_{R}$ to be $\varepsilon$. As such, the number of steps, $M$, required
by Algorithm~\ref{alg:basic-grad-asc-free-energy} then should satisfy the
following:
\begin{equation}
M = \left\lceil 8 \left(\frac{ R r }{\varepsilon}\right)^2 \ln (d+1)
\sum_{i \in[c]} \left\Vert A_{i}\right\Vert ^{2}\right\rceil ,
\end{equation}
which follows from the substitutions $\varepsilon\to\sfrac{\varepsilon}{R}$ and $d \to
d+1$ in~\eqref{eq:number-of-steps-bound-smoothness} and where we used the fact
that $\left\Vert A_{i}\right\Vert = \left\Vert A_{i}^{\prime}\right\Vert $.
Additionally, $r$ is an upper bound on $\left \| \mu^*\right\|$, where $\mu^{*}$ is the optimal solution to the dual problem
in~\eqref{eq:dual-free-energy}, at temperature $T = \frac{\varepsilon}
{4 R\ln(d+1)}$
and under the identifications in~\eqref{eq:identify-SDP-energy-min}. See Figure~\ref{fig:general-SDP} for a concise depiction of the proposed method for solving general SDPs.

\section{Hybrid quantum--classical algorithms for energy minimization and
semi-definite optimization}

\label{sec:q-algs-SDPs}

\subsection{Overview}

Building on the developments in the previous sections,
we now put forward hybrid quantum--classical (HQC) algorithms for energy minimization and
semi-definite optimization. An HQC algorithm for energy minimization is
simpler than that for a general SDP because
Algorithm~\ref{alg:basic-grad-asc-free-energy} maps over essentially directly,
with only two key modifications needed for execution on a quantum computer:

\begin{enumerate}
\item Each element of $\mathcal{Q}$ in~\eqref{eq:observable-tuple} corresponds
to an observable that can be measured on a quantum computer.

\item The parameterized thermal state $\rho_{T}(\mu)$ in~\eqref{eq:grand-canon-thermal-state} can be prepared on a
quantum computer.
\end{enumerate}

\noindent These requirements were already identified in~\cite{Brandao2017},
and these or related requirements have played a key role in several quantum
algorithms proposed for semi-definite optimization
\cite{Brandao2019,Apeldoorn2019,vanApeldoorn2020quantumsdpsolvers,watts2023quantum}. These requirements are also the same as those employed in various quantum
Boltzmann machine learning algorithms~\cite{Anshu2021,Coopmans2024}.

An important aspect of our HQC algorithm for energy minimization is that
the quantum computer is used exclusively for estimating each expectation
$\langle Q_{i}\rangle_{\rho_{T}(\mu)}$, for all $i \in[c]$, during the
gradient update and final steps and then for estimating the expectation $\langle
H - \mu\cdot Q\rangle_{\rho_{T}(\mu)}$ at the final step of
Algorithm~\ref{alg:basic-grad-asc-free-energy}. This requires the ability to
repeatedly prepare the parameterized thermal state $\rho_{T}(\mu)$ and to
measure the observables in $\mathcal{Q}$. As such, it is a variational quantum
algorithm~\cite{Cerezo2021vqa,bharti2021noisy}, i.e., employing a hybrid quantum-classical approach. 
However, in
contrast to many existing variational quantum algorithms,
Algorithm~\ref{alg:basic-grad-asc-free-energy} is guaranteed to converge to an
approximate globally optimal solution.

\subsection{Setup and assumptions}

In what follows, we assume that $d=2^{n}$, so that each observable in
$\mathcal{Q}$ acts on $n$ qubits. We also assume that each
observable in $\mathcal{Q}$ defined in~\eqref{eq:observable-tuple} is efficiently measurable on a quantum computer.
This is the case, for example, if each observable in $\mathcal{Q}$  can
be written as a succinct linear combination of tensor products of Pauli matrices, i.e., if $H$ and every $Q_i$ can be written as
\begin{align}
    H & = \sum_{\vec{\jmath}} h_{\vec{\jmath}} \, \sigma_{\vec{\jmath}} ,\label{eq:H_pauli}\\
    Q_i & = \sum_{\vec{\jmath}} a_{i, \vec{\jmath}} \, \sigma_{\vec{\jmath}} \label{eq:Q_pauli},
\end{align}
for which the number of non-zero terms in each sum is
polynomial in $n$. In the above, we have adopted the following notation,
including the multi-index
\begin{equation}
\vec{\jmath} \coloneqq(j_1, \dots, j_n) \in \{0,1,2,3\}^n,
\end{equation}
with Pauli matrices
$
\sigma_{0} \coloneqq I, \quad  \sigma_1 \coloneqq X, \quad \sigma_2 \coloneqq Y, \quad \sigma_3\coloneqq Z,
$
Pauli string
\begin{equation}
\sigma_{\vec\jmath}\coloneqq 
\sigma_{j_{1}}\otimes\sigma_{j_{2}}\otimes\cdots\otimes\sigma_{j_{n}},
\end{equation}
and coefficients $h_{\vec\jmath},a_{i,\vec\jmath}\in\mathbb R$.

By absorbing any sign into $\sigma_{\vec\jmath}$, we assume, without loss of generality, that all coefficients are non-negative. 
Furthermore,  the triangle inequality for the operator norm and every Pauli string having operator norm equal to one imply that
\begin{align}
    \left\Vert H \right\Vert & \leq \sum_{\vec{\jmath}} \left\vert h_{\vec{\jmath}} \right\vert
\left\Vert \sigma_{\vec{\jmath}} \right\Vert = \sum_{\vec{\jmath}} \left\vert h_{\vec{\jmath}}\right\vert
    \eqqcolon \left\Vert h \right\Vert _{1},\\
\left\Vert Q_i \right\Vert & \leq \sum_{\vec{\jmath}} \left\vert a_{i, \vec{\jmath}} \right\vert
\left\Vert \sigma_{\vec{\jmath}} \right\Vert = \sum_{\vec{\jmath}} \left\vert a_{i, \vec{\jmath}}\right\vert
\eqqcolon \left\Vert a_i \right\Vert _{1}.\label{eq:upper_bound_Q_norm}
\end{align}
Thus, each observable $H$ and $Q_i$ being efficiently measurable on a quantum computer is a consequence of $\left\Vert h \right\Vert _{1}$ and $\left\Vert a_i \right\Vert _{1}$ being polynomial in $n$; i.e., $\left\Vert h \right\Vert _{1} = O(\operatorname{poly}(n))$ and $\left\Vert a_i \right\Vert _{1} = O(\operatorname{poly}(n))$, for all $i \in [c]$. 

\subsection{Approach}

\label{sec:SGA_approach}

Now we present a stochastic gradient-based approach for solving the energy minimization problem in~\eqref{eq:energy-minimization}. We do so, as before, by appealing to the related chemical potential maximization in~\eqref{eq:dual-free-energy}. The details are presented in Algorithm~\ref{alg:energy-min} listed below. At a high level,  Algorithm~\ref{alg:energy-min} proceeds similarly to Algorithm~\ref{alg:basic-grad-asc-free-energy}. The main distinction is that Algorithm~\ref{alg:energy-min} employs stochastic gradient ascent rather than gradient ascent, which means that calculations of $\langle Q_i \rangle_{\rho_T(\mu^{m-1})}$ and $\langle H  - \mu^M\cdot
 Q\rangle _{\rho_{T}(\mu^M)}$ are replaced with respective estimates $\tilde{Q}_i$ and $\tilde{G}$ of these quantities obtained by a quantum computer. These estimates are realized by calling the subroutine $\mathtt{estimate\_obs}$ in steps~7 and 14 of Algorithm~\ref{alg:energy-min}, respectively. See Algorithm~\ref{alg:est_exp_Q} in Appendix~\ref{app:grad_est} for details of this subroutine. Other changes involve using the projection $\Pi_{\mathcal{X}}$ in step~10, the  choice of $\overline{\mu^M}$ in steps~12--15 rather than $\mu^M$, and choices for the learning rate $\eta$ and number of iterations, $M$, in~\eqref{eq:stepsize-SGA} and~\eqref{eq:num_iter_SGA}, respectively. These last changes are made in order to obtain a convergence guarantee for Algorithm~\ref{alg:energy-min}, by invoking a known result on the convergence of stochastic gradient ascent~\cite[Theorem~6.3]{Bubeck2015}.

 Our HQC algorithm for energy minimization is presented as Algorithm~\ref{alg:energy-min} below:

\begin{algorithm}[H]
\caption{$\mathtt{minimize\_energy}(H, (Q_i)_i, (q_i)_i, d, \varepsilon, \delta,r)$}\label{alg:energy-min}
\begin{algorithmic}[1]
\STATE \textbf{Input:}
\begin{itemize}\setlength\itemsep{-0.15em}
    \item Observables $H$ and  $(Q_i)_{i\in [c]}$ (as given in~\eqref{eq:H_pauli}--\eqref{eq:Q_pauli})
    \item Constraint values $(q_i)_{i\in[c]}$
    \item Hilbert space dimension $d$
    \item Accuracy $\varepsilon > 0$
    \item Error probability $\delta\in(0,1)$
    \item Radius $r$: An upper bound on $\|\mu^*\|$, where $\mu^*$ is the optimal solution to~\eqref{eq:dual-free-energy} for $T=\frac{\varepsilon}{4\ln d}$
\end{itemize}

\STATE Initialize $\mu^0 \leftarrow (0, \ldots, 0)$
\STATE Set learning rate $\eta$ as in~\eqref{eq:stepsize-SGA}
\STATE Set number of iterations, $M$, as in~\eqref{eq:num_iter_SGA}
\FOR{$m = 1$ to $M$}
\FOR{$i = 1$ to $c$}
  \STATE $\tilde{Q}_i \leftarrow \mathtt{estimate\_obs}(\mu^{m-1},  (a_{i,\vec\jmath})_{\vec\jmath}, \varepsilon, \delta)$
  \STATE $\overline{g}_i(\mu^{m-1}) \leftarrow q_i - \tilde{Q}_i $ 
  \ENDFOR
  \STATE Update: $\mu^{m} \leftarrow \Pi_{\mathcal{X}} (\mu^{m-1} + \eta \, \overline{g}(\mu^{m-1}))$
\ENDFOR
\STATE Set $\overline{\mu^M} \leftarrow \frac{1}{M}\sum_{m=1}^M \mu^m$
\STATE Set $g_{\vec\jmath} \leftarrow h_{\vec\jmath} - \sum_{i\in[c]}\left[\overline{\mu^M} \right]_i a_{i,\vec\jmath}$, for all $\vec\jmath$
\STATE  $\tilde{G} \leftarrow \mathtt{estimate\_obs}(\overline{\mu^M}, ( g_{\vec\jmath})_{\vec\jmath},  \sfrac{\varepsilon}{4}, \delta)$
\RETURN Output $\overline{\mu^M} \cdot q + \tilde{G}  $

\end{algorithmic}
\end{algorithm}

Let us now discuss various aspects of Algorithm~\ref{alg:energy-min} in more detail. 
Since the objective function $f(\mu)$ defined in~\eqref{eq:obj_fun} is concave (see Lemma~\ref{lem:concavity}), gradient-based methods are guaranteed to converge to the global optimum. However, evaluating the exact gradient $\nabla_\mu f(\mu)$ requires perfect estimates of the expectation values $\left\langle Q_{i} \right\rangle_{\rho_{T}(\mu)}$, and this is not possible when sampling a finite number of times. To address this issue, we employ a projected stochastic gradient ascent algorithm (abbreviated as SGA), in which the gradient is estimated using measurements on quantum states, introducing inherent randomness. See~\cite[Section~3 and page~331]{Bubeck2015} for background on projected stochastic gradient ascent. As we note below, due to the concavity of $f(\mu)$, SGA converges to an 
$\varepsilon$-approximate optimal solution.

The canonical SGA algorithm solves the chemical potential maximization problem by iteratively updating the chemical potential parameter vector $\mu$ as follows:
\begin{equation}\label{eq:SGA}
    \mu^{m+1} = \Pi_{\mathcal{X}}(\mu^m  + \eta \overline{g}(\mu^m)),
\end{equation}
where the superscripts $m+1$ and $m$ on $\mu$ indicate the iteration of the algorithm, $\eta>0$ is the learning rate, and $\overline{g}(\mu^{m})$ is the
stochastic gradient evaluated at $\mu^{m}$. Additionally,  $\mathcal{X} \subset \mathbb{R}^c$ denotes a Euclidean ball of radius~$r$ that contains the optimal solution $\mu^*$  to~\eqref{eq:dual-free-energy} for $T=\frac{\varepsilon}{4\ln d}$, so that
\begin{equation}
    r \geq \left \| \mu^* \right\|.
\end{equation}
Similar to how it was with Algorithm~\ref{alg:basic-grad-asc-free-energy}, we need to have a guess on the size of the search space in $\mathbb{R}^c$ (as given by $r$), and we restrict the search to this space. Also, $\Pi_{\mathcal{X}}$ is a Euclidean projection operation that projects the update $\mu^m  + \eta \overline{g}(\mu^m)$ back to $\mathcal{X}$ if it falls outside of $\mathcal{X}$. Formally, it is defined as
\begin{equation}
    \Pi_{\mathcal{X}}(y)\coloneqq \operatorname{argmin}_{x\in\mathcal{X}
}\left\Vert y-x\right\Vert _{2},
\end{equation}
and it is easy to perform this optimization.
The stochastic
gradient~$\overline{g}(\mu)$ should be unbiased, meaning that
\begin{equation}
\mathbb{E}[\overline{g}(\mu)]=\nabla_{\mu}f(\mu) = q-\left\langle Q\right\rangle _{\rho_{T}(\mu)},
\label{eq:unbiased-condition}
\end{equation}
for all $\mu
\in\mathbb{R}^{c}$, where the expectation is over all the randomness
associated with the generation of~$\overline{g}(\mu)$. Also, the notation $q-\left\langle Q\right\rangle _{\rho_{T}(\mu)}$ is a shorthand for the vector with $i$th component given by $q_{i}-\left\langle Q_{i}\right\rangle _{\rho_{T}(\mu)}$, for all $i\in[c]$.

The gradient of $f(\mu)$ is given in~\eqref{eq:gradient-formula} and \eqref{eq:unbiased-condition}; its elements are as follows:
\begin{equation}
\frac{\partial f(\mu)}{\partial\mu_{i}} = q_{i}-\left\langle Q_{i}\right\rangle _{\rho_{T}(\mu)}.
\end{equation}
To estimate the partial derivative $\frac{\partial f(\mu)}{\partial\mu_{i}}$, we prepare the state $\rho_T (\mu)$ and then measure the observable $Q_{i}$. When $Q_i$ is expressed as a linear combination of Pauli strings, as in~\eqref{eq:Q_pauli}, the measurement can be performed using standard techniques. Through repetition, the estimate of $\left\langle Q_{i}\right\rangle_{\rho_{T}(\mu)} $ can be made as
accurate as desired. This procedure is a subroutine called $\mathtt{estimate\_obs}$ and is described in detail in Appendix~\ref{app:grad_est} as Algorithm~\ref{alg:est_exp_Q}, the result of which is that $O(\left\Vert
a_i \right\Vert _{1}^{2}\varepsilon^{-2}\ln\delta^{-1})$ samples of
$\rho_T(\mu)$ are required to estimate $\left\langle Q_{i}\right\rangle_{\rho_{T}(\mu)} $ with an accuracy of $\varepsilon>0$ and a
failure probability of $\delta\in\left(  0,1\right)  $ (see Appendix~\ref{app:sample_complexity_gradient} for further details).

Finally, we set the number of iterations, $M$, and stepsize $\eta$ of SGA as follows:
\begin{align}
\eta & \coloneqq \left[  \frac{8\ln d}{\varepsilon}
\sum_{i\in\left[  c\right]  }\left\Vert a_{i}\right\Vert _{1}^2+\frac{\sigma}
{r}\sqrt{\frac{M}{2}}\right]  ^{-1},\label{eq:stepsize-SGA}\\
M& \coloneqq \left\lceil \frac{16r^2}{\varepsilon^{2}}\left(  2\sigma^{2}+8 ( \ln
d)  \sum_{i\in\left[  c\right]  }\left\Vert a_{i}\right\Vert
_{1}^2\right)  \right\rceil ,\label{eq:num_iter_SGA}\\
\sigma^2 & \coloneqq c\varepsilon
^{2}+\delta\sum_{i\in\left[  c\right]  }\left\Vert a_{i}\right\Vert _{1}^{2}.
\end{align}
These choices are made to guarantee convergence of Algorithm~\ref{alg:energy-min} and are justified in Appendix~\ref{app:sample-comp-SGD}.

\subsection{Convergence of hybrid quantum--classical algorithm for energy minimization}

By invoking~\cite[Theorem~6.3]{Bubeck2015} and further analysis in Appendix~\ref{app:sample-comp-SGD}, we conclude that SGA  converges to an $\varepsilon$-approximate globally optimum solution with a sample complexity given in~\eqref{eq:main-sample-complexity}. That is, our  convergence guarantee for Algorithm~\ref{alg:energy-min} and the energy minimization problem is as follows:

\begin{theorem}
\label{thm:main}
Let $H$ and $Q_{i}$ for each
$i\in\lbrack c]$, $c\in\mathbb{N}$, be a Hamiltonian and observables as
defined in~\eqref{eq:H_pauli} and~\eqref{eq:Q_pauli} respectively, and let $h$ and
$a_{i}$ be the Pauli-coefficient vectors of $H$ and $Q_{i}$, respectively. Let $\varepsilon>0$ and $ \delta\in (0,1)$. Then
to reach, in expectation, an
$\varepsilon$-approximate estimate $\hat{E}$ of the optimal value
in~\eqref{eq:energy-minimization} with success probability $\geq1-\delta$, i.e., such that
\begin{equation}
\Pr\!\left[\left\vert E(\mathcal{Q},q)-\mathbb{E}\!\left[\hat{E}\right]\right\vert \leq\varepsilon\right] \geq 1-\delta ,
\label{eq:SGA-guarantee-main-text}
\end{equation}
the sample complexity (the number of thermal-state samples  of the form in~\eqref{eq:grand-canon-thermal-state}) of Algorithm~\ref{alg:energy-min} is given by 
\begin{equation}
    O\!\left (\frac{r^2 \ln\!\left(\sfrac{1}{\delta}\right)\ln d}{\varepsilon^4}  \left[\max\!\left\{\sum\nolimits_{i\in[c]} \left \Vert a_i \right \Vert_1^2 , \left \Vert h \right \Vert_1 \right\}\right]^2\right).
    \label{eq:main-sample-complexity}
\end{equation}
In~\eqref{eq:SGA-guarantee-main-text}, the expectation $\mathbb{E}$ is with respect to the $M$ steps of SGA updates, and the probability $\Pr$ is with respect to the final step in estimating
\begin{equation}
    \overline{\mu^{M}}\cdot   q + \langle H - \overline{\mu^{M}}\cdot
Q\rangle _{\rho_{T}(\overline{\mu^{M}})},
\end{equation}
where $\overline{\mu^M} \coloneqq \frac{1}{M}\sum_{m=1}^M \mu^m$.
Additionally, the radius $r$ is such that $r\geq \left \| \mu^*\right\|$, where $\mu^*$ satisfies $f(\mu^{\ast})   =F_{T}(\mathcal{Q},q)$ for $T = \frac{\varepsilon}{4\ln d}$ and $f(\mu)$ defined in~\eqref{eq:obj_fun}.
\end{theorem}

\begin{proof}
    See Appendix~\ref{app:sample-comp-SGD}.
\end{proof}

\medskip
Note that $\overline{\mu^{M}} $ in Theorem~\ref{thm:main} is a random variable, as
specified in Algorithm~\ref{alg:energy-min}, because it is a function of the stochastic gradients $\overline{g}(\mu^{M})$, \ldots, $\overline{g}(\mu^{0})$. Thus, as mentioned above, the expectation in~\eqref{eq:SGA-guarantee-main-text} is
with respect to the randomness associated with generating~$\overline{\mu^{M}}$.

After generating $\overline{\mu^{M}}$, the final step of Algorithm~\ref{alg:energy-min} involves estimating the expectation $\langle H - \overline{\mu^{M}}\cdot
Q\rangle _{\rho_{T}(\overline{\mu^{M}})}$ by Algorithm~\ref{alg:est_exp_Q} in Appendix~\ref{app:grad_est}. There is further randomness involved in this step, and the probability $\Pr$ in \eqref{eq:SGA-guarantee-main-text} is with respect to this randomness. Note that the number of trials needed to arrive at a good estimate in this last step is independent of the particular value of $\overline{\mu^{M}}$.

In summary, Theorem~\ref{thm:main} states that the sample complexity of  Algorithm~\ref{alg:energy-min} is polynomial in $\varepsilon^{-1}$, $c$, $r$, $\left\|h\right\|_1$, and $\left\| a_i\right  \|_1$ for all $i\in [c]$. Additionally, for $c = \operatorname{poly}(n)$, $r=\operatorname{poly}(n)$, $\left\|h\right\|_1 = \operatorname{poly}(n)$, and $\left\|a_i\right\|_1 = \operatorname{poly}(n)$ for all $i\in[c]$, the sample complexity scales polynomially with number of qubits, $n$, and inverse accuracy~$\varepsilon^{-1}$.

\subsection{Extension to semi-definite optimization}

In light of the mapping of a general SDP in~\eqref{eq:SDP-standard} to an energy minimization problem presented in Section~\ref{sec:general-SDPs}, Algorithm~\ref{alg:energy-min} extends to solving SDPs. The following lemma formalizes this reduction, which is more amenable to implementation on a quantum computer than the reduction from Lemma~\ref{lem:reduction-SDP-to-energy-min}:

\begin{lemma}
\label{lem:reduction-SDP-to-energy-min_quantum} The following
equality holds:
\begin{equation}
\alpha_{R}=R\cdot\min_{\rho\in\mathcal{D}_{2d}}\left\{  \left\langle
C^{\prime}\right\rangle _{\rho}:\left\langle A_{i}^{\prime}\right\rangle
_{\rho}=b_{i}^{\prime}\ \forall i\in\left[  c\right]  \right\}  ,
\label{eq:reduction-to-states_quantum}
\end{equation}
where
\begin{equation}
C'\coloneqq C \otimes |0\rangle\!\langle 0|  ,\ A_i' \coloneqq A_i \otimes |0\rangle\!\langle 0|  ,\ b_{i}^{\prime} \coloneqq \frac
{b_{i}}{R} \quad\forall i\in\left[  c\right]  .
\end{equation}

\end{lemma}

\begin{proof}
See Appendix~\ref{app:reduction-SDP-to-energy-min_quantum}.
\end{proof}

\medskip

Given the reduction from Lemma~\ref{lem:reduction-SDP-to-energy-min_quantum}, it is thus possible to solve SDPs of the general form in~\eqref{eq:SDP-standard} by means of Algorithm~\ref{alg:energy-min}. Indeed, because this mapping preserves the concavity of the dual objective, the HQC approach and the algorithm's update rule remain unchanged. The only substantial modification lies in the error tolerance: to guarantee an overall error of $\varepsilon$ for the original SDP, the algorithm now accounts for the trace-bound scaling $R$ introduced in the mapping by using an error $\sfrac{\varepsilon}{R}$. Thus, the convergence and sample complexity guarantees in Theorem~\ref{thm:main} carry over directly for solving general SDPs with SGA, with the sample complexity in~\eqref{eq:main-sample-complexity} generalized by accounting for a scaled value of $\varepsilon \to \sfrac{\varepsilon}{R}$. Moreover, the mapping in Lemma~\ref{lem:reduction-SDP-to-energy-min_quantum} doubles the dimension of the system, i.e., $d \to 2d$, which in the case of $n$ qubits yields the change $2^n \to 2^{n+1}$, because $d=2^n$.  
Overall, the sample complexity in~\eqref{eq:main-sample-complexity}
incurs only polynomial overhead from the trace factor $R$ and the single-qubit extension of the system. Specifically, the sample complexity for solving SDPs of the general form in~\eqref{eq:SDP-standard} by means of Algorithm~\ref{alg:energy-min} is given by
\begin{equation}
    O\!\left (\frac{r^2R^4\, n\ln\!\left(\sfrac{1}{\delta}\right)}{\varepsilon^4}  \left[\max\left\{\sum\nolimits_{i\in[c]} \left \Vert A_i \right \Vert^2 , \left \Vert C \right \Vert \right\}\right]^2\right).
    \label{eq:main-sample-complexity_SDP}
\end{equation}

\section{Second-order stochastic gradient ascent}

\label{sec:second-order}

\begin{figure*}
    \centering
    \scalebox{1.79}{
    \Qcircuit @C=1.0em @R=1.0em {
            \lstick{|1\rangle} & \gate{\operatorname{Had}} & \ctrl{1} &  \gate{\operatorname{Had}} & \meter & \rstick{\hspace{-1.2em}Z} \\
            \lstick{\rho_T(\mu)} & \qw & \gate{\sigma_{\vec{k}}} & \gate{e^{i\left(  H-\mu\cdot Q\right)  t/T}} &   \meter & \rstick{\hspace{-1.2em}\sigma_{\vec{\ell}}}
        }
    }
    \caption{Quantum circuit that realizes an unbiased estimate of $-\frac{1}{2} \left\langle 
    \left\{  \Phi_{\mu}(Q_{i}),Q_{j}\right\}
    \right\rangle _{\rho_{T}(\mu)} $. For each run of the circuit, the time $t$ is sampled independently at random from the probability density $p(t)$ in~\eqref{eq:high-peak-tent-density-def}, while $\sigma_{\vec{\ell}}$ and $\sigma_{\vec{k}}$ are sampled with probabilities $\sfrac{a_{i,\vec{\ell}}}{\left \Vert a_i \right \Vert_1}$ and $\sfrac{a_{j,\vec{k}}}{\left \Vert a_j \right \Vert_1}$, respectively. For details of the algorithm, see Appendix~\ref{app:Hessian_est}.}
        \label{fig:KM}
\end{figure*}

In Section~\ref{sec:SGA_approach}, we presented a stochastic gradient-based approach for solving the chemical potential maximization problem in~\eqref{eq:dual-free-energy} and, with Algorithm~\ref{alg:est_exp_Q}, showed how to obtain an unbiased estimator of the stochastic gradient of the objective function $f(\mu)$. Furthermore, in Section~\ref{sec:algorithm-analysis} we observed that the matrix elements of the Hessian of $f(\mu)$ admit the closed form in~\eqref{eq:hessian_via_KM}. Here we provide an equivalent expression that is more convenient for estimation on a quantum computer. To do so, let us first recall the following quantum channel:
\begin{equation}
\Phi_{\mu}(\cdot)\coloneqq \int_{-\infty}^{\infty}dt\ p(t)\, e^{-i\left(  H-\mu\cdot
Q\right)  t/T}(\cdot)e^{i\left(  H-\mu\cdot Q\right)  t/T},\label{eq:Phi}
\end{equation}
where
\begin{equation}
p(t)\coloneqq\frac{2}{\pi}\ln\!\left\vert \coth\!\left(  \frac{\pi t}{2}\right)
\right\vert 
\label{eq:high-peak-tent-density-def}
\end{equation}
is a probability density function known as the \textit{high-peak-tent} density~\cite{patel2024quantumboltzmannmachinelearning}. A detailed account of the quantum channel in~\eqref{eq:Phi} and its use in different contexts can be found in~\cite[Section II]{minervini2025evolvedquantumboltzmannmachines} (see also \cite{Hastings2007,Kim2012,Anshu2021,Coopmans2024,patel2024quantumboltzmannmachinelearning,patel2024naturalgradientparameterestimation}).
Using this definition,  the matrix elements of the Hessian of $f(\mu)$ can be written as follows:
\begin{multline}\label{eq:hessian_2}
    \frac{\partial^{2}}{\partial\mu_{i}\partial\mu_{j}}f(\mu)
    = \frac{1}{T}\Big(  
    \left\langle Q_{i}\right\rangle _{\rho_{T}(\mu)}
    \left\langle Q_{j}\right\rangle _{\rho_{T}(\mu)} \\
     -\frac{1}{2}\left\langle 
    \left\{  \Phi_{\mu}(Q_{i}),Q_{j}\right\}
    \right\rangle _{\rho_{T}(\mu)}\Big).
\end{multline}
See Lemma~\ref{lem:hessian-kubo-mori} in Appendix~\ref{app:gradient-hessian} for a derivation.

As already stated in~\eqref{eq:hess_-KM}, the expression in~\eqref{eq:hessian_2} is equal to the negative of the Kubo--Mori (KM) information matrix of the parameterized family $\left(
\rho_{T}(\mu)\right)  _{\mu\in\mathbb{R}^{c}}$ of thermal states; that is,
\begin{equation}
\frac{\partial^{2}}{\partial\mu_{i}\partial\mu_{j}}f(\mu)=-I_{ij}
^{\operatorname{KM}}(\mu). \label{eq:hessian-kubo-mori_2}
\end{equation}
As such, the Hessian elements in~\eqref{eq:hessian_2} can be estimated using methods similar to those proposed recently for estimating the Kubo--Mori information matrix elements of quantum Boltzmann machines (see~\cite[Theorem 2]{patel2024naturalgradientparameterestimation}). In particular, a detailed algorithm for estimating the quantity in~\eqref{eq:hessian_2} is recalled in Appendix~\ref{app:Hessian_est}, and the quantum circuit that plays a role in estimating the second term on the right-hand side of~\eqref{eq:hessian_2} is shown in Figure~\ref{fig:KM}.

The ability to estimate Hessian elements enables the implementation of a second-order method, often referred to as Newton's method, for optimizing $f(\mu)$. Considering that the objective function $f(\mu)$ is smooth and concave, this approach can offer faster local convergence than first-order techniques~\cite{Boyd2004,nocedal2006numerical_opt}. Moreover, the equivalence of the Hessian with the Kubo--Mori information matrix as described in~\eqref{eq:hessian-kubo-mori_2} implies that the Euclidean curvature of $f(\mu)$ coincides with the information-geometric curvature induced by the (Umegaki) relative entropy~\cite{Umegaki1962,sbahi2022provablyefficientvariationalgenerative}. In this sense, a Newton step is identical to a natural-gradient step with respect to the Kubo-Mori metric~\cite{amari1998nat_grad}, where the update rule is given as follows:
\begin{equation}
    \mu^{m+1} = \mu^m  - \eta \Delta_m,
    \label{eq:update-rule-2nd-order}
\end{equation}
where $\eta>0$ is the learning rate, $\overline{g}(\mu^{m})$ is the stochastic gradient already introduced in the SGA in~\eqref{eq:SGA}, $\Delta_m$ is the solution to
\begin{equation}
    I^{\operatorname{KM}}(\mu^m) \Delta_m =   \overline{g}(\mu^m),
    \label{eq:linear-sys-KM}
\end{equation}
and $I^{\operatorname{KM}}(\mu^m)$ is the Kubo--Mori information matrix encoding the curvature of the state space. In \eqref{eq:update-rule-2nd-order}, the minus sign appears due to the equality in~\eqref{eq:hessian-kubo-mori_2}. By solving~\eqref{eq:linear-sys-KM}, which effectively amounts to incorporating the inverse of $I^{\operatorname{KM}}(\mu^m)$, the gradient is rescaled to align with the steepest descent direction in the geometry defined by the quantum relative entropy.

\section{Comparison to matrix multiplicative weights update and matrix
exponentiated gradient update methods}

\label{sec:MEG-MMW-compare}

As we have argued in Section~\ref{sec:energy-min-q-thermo}, some of the basic ideas that connect quantum thermodynamics to semi-definite optimization were already laid out by Jaynes~\cite[page~200]{Jaynes1962}. In particular, he showed that the related problem of entropy maximization can be solved in a manner outlined in steps~2 and 3 of Section~\ref{sec:summary-contribs}, by considering a dual optimization problem and solving that by gradient ascent. Since Jaynes' contribution, along a different line, computer scientists and operations researchers proposed related algorithms, one called the matrix exponentiated gradient (MEG) update method~\cite{Tsuda2005} and another called the matrix multiplicative weights (MMW) update method~\cite{Arora2007,Arora2012,Arora2016}.

Both of these approaches have some relations to the Jaynes-inspired approach put forward in our paper. In the MEG method, certain optimization problems are cast in the terminology of Bregman divergences and derivations are provided showing that, in the physics terminology used here, parameterized thermal states are optimal for solving such optimization problems. However, the MEG method is not motivated in~\cite{Tsuda2005} from a quantum thermodynamics perspective. The MMW method is motivated as a matrix generalization of the multiplicative weights update method for online learning and linear optimization, and derives from considerations in the field of zero-sum games. Furthermore, this approach involves updates of what amount to parameterized thermal states, but little motivation is given in~\cite{Arora2007} as to why one should choose these particular states as an ansatz, and there is again no connection made to quantum thermodynamics.

The different motivations coming from game theory and thermodynamics does affect the way that the algorithms for semi-definite programs are explicitly implemented. For example, since MMW is motivated from zero-sum games, it is more natural to consider the feasibility problem, as detailed in \cite{Arora2007} and all previous quantum algorithms for solving SDPs \cite{Brandao2017,vanApeldoorn2020quantumsdpsolvers,Brandao2019,Apeldoorn2019}, rather than the original formulation of the SDP problem, and the iterative updates in MMW correspond to binary search. On the other hand, the thermodynamic consideration implies that successive updates to the chemical potential are more natural (the direction depending on whether or not the total charge is smaller or larger than the given expected total), and this corresponds to gradient ascent. We plan to explore the deeper relationship between MMW, MEG, and quantum thermodynamics in future work. 

We now compare the sample complexities of Algorithm~\ref{alg:energy-min} for solving SDPs with the current state-of-the-art quantum algorithm based on the MMW approach~\cite{Apeldoorn2019}, hereafter referred to as AG19. However, before proceeding, we first align the problem settings considered in our work and in~\cite{Apeldoorn2019} for a fair comparison. In~\cite{Apeldoorn2019}, the authors assumed that that the norms of the input matrices $C, A_{1}, \ldots, A_{c}$ are bounded from above by one, that is, $\left \Vert C \right \Vert \leq 1$ and $\left \Vert A_{i} \right \Vert \leq 1$ for all $i\in [c]$, whereas we do not make any such assumptions. As a result, for the purpose of this section, we make this assumption in order to provide a fair comparison.

Our analysis reveals that our algorithm (Algorithm~\ref{alg:energy-min}), though conceptually simple, achieves better sample complexity than that of~\cite{Apeldoorn2019} in key regimes. Specifically, we see an improvement in the dependence on the parameter $r$: the sample complexity of Algorithm~\ref{alg:energy-min} for solving SDPs (see~\eqref{eq:main-sample-complexity_SDP}) scales as
\begin{equation}
O\!\left(\frac{r^2 c^2 R^4  \ln\! \left (\sfrac{1}{\delta}\right) \ln d}{\varepsilon^4}\right)
\end{equation}
while the sample complexity of the algorithm of~\cite{Apeldoorn2019} scales as 
\begin{equation}
    O\!\left(\frac{r^4 R^4  \ln\! \left (\sfrac{1}{\delta}\right)\ln d}{\varepsilon^4}\right)\label{eq:sample-comp-prior-method}.
\end{equation}
See \cite[Theorem~8]{vanApeldoorn2019arXiv}  and Appendix~\ref{app:samp-comp-of-prior-work} for more details on~\eqref{eq:sample-comp-prior-method}.

The comparison reveals an interesting trade-off: while our algorithm (Algorithm~\ref{alg:energy-min}) demonstrates a quadratic improvement in the dependence on $r$, the sample complexity of the quantum MMW algorithm of~\cite{Apeldoorn2019} is independent of the parameter $c$. This suggests that for combinatorial problems like MAXCUT, where $r$ scales linearly with $c$, both algorithms achieve comparable performance. However, in cases where $r \gg c $, our algorithm provides a strictly better sample complexity, making it particularly advantageous for problems with large dual radius.

\section{Conclusion}

\label{sec:conclusion}

\subsection{Summary}

In this paper, we showed how physical intuition from quantum thermodynamics can be used for understanding semi-definite optimization. In particular, we highlighted that energy minimization in the presence of non-commuting conserved charges and those of semi-definite optimization share a common mathematical core. In light of Jaynes' statistical-mechanical perspective, the path we took to show this first involved stepping from the energy minimization problem in~\eqref{eq:energy-minimization} to the low-temperature free-energy minimization in~\eqref{eq:min-free-energy}, whose error is controlled by $T \ln d$, where $T$ is the temperature and $d$ is the dimension of the quantum system. An analysis using Lagrange duality and quantum relative entropy then yielded the chemical-potential maximization problem in~\eqref{eq:dual-free-energy} as being equal to the free energy minimization problem in~\eqref{eq:min-free-energy}. The chemical-potential maximization problem in~\eqref{eq:dual-free-energy} involves a smooth concave objective function, which can be tackled efficiently with first-order gradient ascent and the use of parameterized thermal states of the form in~\eqref{eq:grand-canon-thermal-state}. 

Leveraging this structure, we developed classical and hybrid quantum-classical (HQC) algorithms, Algorithms~\ref{alg:basic-grad-asc-free-energy} and~\ref{alg:energy-min} respectively, for performing energy minimization. These algorithms are augmented with a convergence analysis proving that their outputs are $\varepsilon$-approximate global optima. Our approach clarifies that thermal states of temperature $T = O\!\left(\frac{\varepsilon}{\ln d}\right)$ are required for achieving an $\varepsilon$ error in the approximation. Furthermore, by employing an analytical form for the Hessian of the objective function in~\eqref{eq:obj_fun}, we established a second-order method that incorporates second-derivative information. Interestingly, the Hessian matrix is equal to the negative of the Kubo--Mori information matrix, indicating that the Euclidean geometry on the chemical potential parameter space is aligned with the geometry induced by the quantum relative entropy on the state space realized by parameterized thermal states. 

These classical and HQC algorithms can be generalized straightforwardly to solve general SDPs of the standard form in~\eqref{eq:SDP-standard}, because a general SDP can be mapped to a scaled energy minimization problem, as shown in Lemmas~\ref{lem:reduction-SDP-to-energy-min} and~\ref{lem:reduction-SDP-to-energy-min_quantum}. As a consequence, the convergence properties also hold for solving SDPs up to the scaling factor overhead.

\subsection{Discussion}

The applicability of our framework extends beyond equality-constrained SDPs, as defined in~\eqref{eq:SDP-standard}, to mixed equality and inequality constraints. In the dual formulation in~\eqref{eq:obj_fun}, this extension emerges naturally: chemical potentials corresponding to equality constraints remain unconstrained ($\mu \in \mathbb{R}$), while those for inequality constraints are restricted to non-negative values ($\mu \geq 0$).  This distinction preserves both the dual problem's concavity and its alignment with convex optimization principles, while expanding the method's practical utility.

The convergence guarantees of the classical and HQC algorithms (Algorithm~\ref{alg:basic-grad-asc-free-energy} and~\ref{alg:energy-min}, respectively) are supported by the concavity of the dual objective function in~\eqref{eq:obj_fun}.
These convergence properties become improved when the dual problem exhibits strong concavity, which occurs when the log partition function is strongly convex. Existing results establishing this property assume quasi-local operators~\cite{Anshu2021}, but this assumption does not hold for the general setup we have considered in \eqref{eq:H_pauli}--\eqref{eq:Q_pauli}. The assumption does hold, however, for special cases of \eqref{eq:H_pauli}--\eqref{eq:Q_pauli} in which the Pauli operators with non-zero coefficients act locally on only a few qubits, and for these cases we expect a speedup from strong concavity.

Finally, when executing the HQC algorithm (Algorithm~\ref{alg:energy-min}), while our analysis assumes perfect thermal state preparation, practical implementations should account for imperfect state preparation in quantum devices. This introduces bias in the gradient estimators, as the measured observables no longer exactly correspond to the target thermal state expectations. However, our stochastic gradient framework can naturally accommodate this by modifying the variance bound in~\eqref{eq:bounded-variance-condition} to $\sigma^2 \to \sigma^2 + \delta$, where $\delta$ is a function of the error in thermal-state preparation. This modification preserves the convergence guarantees while realistically accounting for thermal state preparation imperfections.

\subsection{Future directions}

Several promising directions emerge from this work. First, a rigorous analysis of the sample complexity for second-order methods presents a natural next step. The concavity of the cost function, combined with the closed-form expression for the Hessian matrix elements, suggests potential convergence improvements over first-order gradient ascent. However, a detailed analysis is still needed to quantify the trade-offs between computational overhead---stemming from Hessian estimation and solving~\eqref{eq:linear-sys-KM}---and the improved sample efficiency enabled by curvature-aware updates. Such an analysis would delineate the regimes where second-order methods offer a practical advantage. At the least, it seems that there could be a strong benefit from using the second-order method when the number of constraints, $c$, is small because the extra overhead from the second-order method is small in such a case. 

Another promising direction is the generalization of our framework to broader classes of SDPs. While we focused on energy minimization with non-commuting conserved charges, the underlying approach---consisting of duality, parameterized thermal states, and gradient-based optimization---could extend to problems with different constraints or cost functions. For instance, SDPs with inequality constraints or other optimization problems with nonlinear objectives may admit similar reformulations using free energy minimization. Exploring such extensions could further bridge quantum thermodynamics with optimization theory and expand the applicability of the approaches put forward here.

Our analysis showed how the Umegaki relative entropy was instrumental in connecting free energy minimization to the dual chemical potential problem. This connection naturally suggests investigating whether generalized divergences, such as  Rényi relative entropies, could play an analogous role. Success in this direction may offer new insights, and leveraging the broader family of Rényi entropies and their properties could lead to alternative optimization algorithms.

Another important extension is to encompass also continuous-variable constraints, and this could be applied to systems that obey known dynamical equations in terms of ordinary and partial differential equations. 

The training process itself can also be generalised to an analog system running in continuous time instead of discrete steps, and this is deeply connected to dynamical systems like replicator dynamics in the quantum regime, where the latter appears in the area of quantum zero-sum games \cite{jain2022matrix,Lin2025learninginquantum}. Exploring these connections has the potential to tie concepts in quantum thermodynamics, dynamical systems, and game theory via semi-definite programs. Due to the deep ties with matrix multiplicative weight updates, the methods presented can also extend beyond semi-definite programs to the areas of online learning and reinforcement learning. 

Finally, the potential for quantum advantage in the HQC algorithm merits further investigation. A speedup would require a problem class where classical simulation of thermal states is intractable---for example, due to non-local correlations or high entanglement---while quantum computers can efficiently prepare these states at low temperatures. However, the temperature scaling inversely with the number of qubits (as required for accurate energy minimization) poses a challenge, as existing thermal state preparation methods struggle to be efficient at such low temperatures. Identifying specific problem classes where our approach outperforms classical methods, or developing more efficient quantum subroutines for low-temperature thermalization, could help establish the conditions for a practical quantum advantage.

\begin{acknowledgments}
We are grateful to Harry Buhrman and Zo\"e Holmes for helpful discussions.

NL acknowledges funding from the Science and Technology Commission of Shanghai Municipality (STCSM) grant no.~24LZ1401200 (21JC1402900), NSFC grants No.~12471411 and No.~12341104, the Shanghai Jiao Tong University 2030 Initiative, and the Fundamental Research Funds for the Central Universities.  MM, DP, and MMW acknowledge support from the National Science Foundation under Grant No.~2329662.
\end{acknowledgments}

\bibliography{Ref}

\appendix
\onecolumngrid

\section{Proof of Equation~\eqref{eq:energy-free-energy-bounds}}

\label{app:energy-free-energy-bnd}

\begin{lemma}
\label{lem:error-approx-SDP}The following bounds hold for every temperature
$T>0$:
\begin{equation}
E(\mathcal{Q},q)\geq F_{T}(\mathcal{Q},q)\geq E(\mathcal{Q},q)-T\ln d.
\end{equation}

\end{lemma}

\begin{proof}
We prove the left inequality first. Let $\rho$ be feasible for $E(\mathcal{Q}
,q)$; i.e., it is a density matrix that satisfies the conservation of the charges specified in the definition of $E(\mathcal{Q}
,q)$.  Then $\rho$ is also feasible for $F_{T}(\mathcal{Q},q)$. Since the
objective function value decreases by an additive factor of $TS(\rho)$ when
considering $E(\mathcal{Q},q)$ and $F_{T}(\mathcal{Q},q)$, the inequality
$E(\mathcal{Q},q)\geq F_{T}(\mathcal{Q},q)$ follows.

Now we prove the right inequality. Let $\rho$ be feasible for $F_{T}
(\mathcal{Q},q)$. Then it is also feasible for $E(\mathcal{Q},q)$.
Furthermore, the following bound holds:
\begin{equation}
\operatorname{Tr}[H\rho]-TS(\rho)\geq\operatorname{Tr}[H\rho]-T\ln d,
\end{equation}
as a consequence of the well known dimension bound on the von Neumann entropy,
and the dimension of $\rho$ being $d\times d$. As a result, we find that
\begin{align}
F_{T}(\mathcal{Q},q)  &  \geq\min_{\rho\in\mathcal{D}}\left\{
\operatorname{Tr}[H\rho]-T\ln d:\operatorname{Tr}[Q_{i}\rho]=q_{i}\ \forall
i\in\left[  c\right]  \right\} \\
&  =\min_{\rho\in\mathcal{D}}\left\{  \operatorname{Tr}[H\rho
]:\operatorname{Tr}[Q_{i}\rho]=q_{i}\ \forall i\in\left[  c\right]  \right\}
-T\ln d\\
&  =E(\mathcal{Q},q)-T\ln d,
\end{align}
thus concluding the proof.
\end{proof}

\section{Proof of Equation~\eqref{eq:dual-free-energy}}

\label{app:proof-dual-expression}

Recall that
\begin{align}
F_{T}(\mathcal{Q},q)  &  \coloneqq\min_{\rho\in\mathcal{D}_{d}}\left\{
\left\langle H\right\rangle _{\rho}-TS(\rho):\left\langle Q_{i}\right\rangle
_{\rho}=q_{i}\ \forall i\in\left[  c\right]  \right\} , \\
Z_{T}(\mu)  &  \coloneqq\operatorname{Tr}\!\left[  \exp\!\left(  -\frac{1}
{T}\left(  H-\sum_{i\in\left[  c\right]  }\mu_{i}Q_{i}\right)  \right)  \right]  ,\\
\rho_{T}(\mu)  &  \coloneqq\frac{1}{Z_{T}(\mu)}\exp\!\left(  -\frac{1}
{T}\left(  H-\sum_{i\in\left[  c\right]  }\mu_{i}Q_{i}\right)  \right)  . 
\end{align}
Consider the following chain of equalities:
\begin{align}
F_{T}(\mathcal{Q},q)  &  =\min_{\rho\in\mathcal{D}_{d}}\left\{
\operatorname{Tr}[H\rho]-TS(\rho):\operatorname{Tr}[Q_{i}\rho]=q_{i}\ \forall
i\in\left[  c\right]  \right\} \\
&  =\min_{\rho\in\mathcal{D}_{d}}\left\{  \operatorname{Tr}[H\rho
]-TS(\rho)+\sup_{\mu\in\mathbb{R}^{c}}\sum_{i\in\left[  c\right]  }\mu_{i}\left(
q_{i}-\operatorname{Tr}[Q_{i}\rho]\right)  \right\} \\
&  =\min_{\rho\in\mathcal{D}_{d}}\sup_{\mu\in\mathbb{R}^{c}}\left\{  \mu\cdot
q+\operatorname{Tr}[H\rho]-TS(\rho)-\sum_{i\in\left[  c\right]  }\mu_{i}\operatorname{Tr}
[Q_{i}\rho]\right\} \\
&  =\sup_{\mu\in\mathbb{R}^{c}}\min_{\rho\in\mathcal{D}_{d}}\left\{  \mu\cdot
q+\operatorname{Tr}[H\rho]-TS(\rho)-\sum_{i\in\left[  c\right]  }\mu_{i}\operatorname{Tr}
[Q_{i}\rho]\right\}   .\label{eq:FE-proof-steps-first}
\end{align}
The second equality follows by introducing Lagrange multipliers for each
constraint. The fourth equality follows by the Sion minimax theorem for
extended reals~\cite[Theorem~2.11]{BenDavid2023}. Indeed, the set
$\mathcal{D}_{d}$ of $d\times d$ density matrices is convex and compact, the
set $\mathbb{R}^{c}$ is convex, and the objective function is concave in
$\rho$, due to concavity of the von Neumann entropy, and linear in $\mu$.
Now consider that
\begin{align}
& \min_{\rho\in\mathcal{D}_{d}}\left\{  \mu\cdot
q+\operatorname{Tr}[H\rho]-TS(\rho)-\sum_{i\in\left[  c\right]  }\mu_{i}\operatorname{Tr}
[Q_{i}\rho]\right\} \notag \\
&  =T\min_{\rho\in\mathcal{D}_{d}}\left\{
\frac{\mu\cdot q}{T}+\frac{1}{T}\operatorname{Tr}[H\rho]-S(\rho)-\sum
_{i=1}^{c}\frac{\mu_{i}}{T}\operatorname{Tr}[Q_{i}\rho]\right\} \label{eq:FE-proof-steps-1}\\
&  =T\min_{\rho\in\mathcal{D}_{d}}\left\{
\frac{\mu\cdot q}{T}-S(\rho)-\operatorname{Tr}\!\left[  \rho\left( -\frac
{1}{T}\left(  H-\sum_{i\in\left[  c\right]  }\mu_{i}Q_{i}\right)  \right)  \right]  \right\}
\\
&  =T\min_{\rho\in\mathcal{D}_{d}}\left\{
\frac{\mu\cdot q}{T}-S(\rho)-\operatorname{Tr}\!\left[  \rho\ln\exp\!\left(
-\frac{1}{T}\left(  H-\sum_{i\in\left[  c\right]  }\mu_{i}Q_{i}\right)  \right)  \right]
\right\} \\
&  =T\min_{\rho\in\mathcal{D}_{d}}\left\{
\frac{\mu\cdot q}{T}-S(\rho)-\operatorname{Tr}\!\left[  \rho\ln\rho_{T}
(\mu)\right]  -\ln Z_{T}(\mu)\right\} \\
&  =T\min_{\rho\in\mathcal{D}_{d}}\left\{
\frac{\mu\cdot q}{T}+D(\rho\Vert\rho_{T}(\mu))-\ln Z_{T}(\mu)\right\} \\
&  =  \mu\cdot q-T\ln Z_{T}(\mu). \label{eq:FE-proof-steps-last}
\end{align}
 The
penultimate equality follows by recognizing the (Umegaki)\ relative entropy
\cite{Umegaki1962}, defined for states $\omega$ and $\tau$ as $D(\omega
\Vert\tau)\coloneqq\operatorname{Tr}[\omega\left(  \ln\omega-\ln\tau\right)
]$. Recall that
\begin{equation}
D(\omega\Vert\tau)\geq0
\end{equation}
for all states $\omega$ and $\tau$, and the inequality is saturated if and
only if $\omega=\tau$. Applying this observation leads to the last equality,
indicating that $\rho=\rho_{T}(\mu)$ is the optimal choice for $\rho$. Additionally, we have shown that
\begin{align}
    \mu\cdot q-T\ln Z_{T}(\mu) & = \mu\cdot
q+\operatorname{Tr}[H\rho_{T}(\mu)]-TS(\rho_{T}(\mu))-\sum_{i\in\left[  c\right]  }\mu_{i}\operatorname{Tr}
[Q_{i}\rho_{T}(\mu)] \label{eq:log-part-back-to-free-energy-1}\\
& = \mu\cdot
q+\langle H\rangle_{\rho_{T}(\mu)}-TS(\rho_{T}(\mu))-\sum_{i\in\left[  c\right]  }\mu_{i}\langle Q_{i}\rangle _{\rho_{T}(\mu)}\label{eq:log-part-back-to-free-energy-2}\\
& = \mu\cdot
q+ \langle H - \mu \cdot Q \rangle_{\rho_{T}(\mu)} -TS(\rho_{T}(\mu)).
\end{align}
Plugging~\eqref{eq:FE-proof-steps-1}--\eqref{eq:FE-proof-steps-last} back into~\eqref{eq:FE-proof-steps-first}, we conclude that
\begin{equation}
    F_{T}(\mathcal{Q},q) = \sup_{\mu\in\mathbb{R}^{c}}\left\{  \mu\cdot q-T\ln Z_{T}(\mu)\right\}.
\end{equation}

\section{Gradient and Hessian of Equation~\eqref{eq:dual-free-energy}}

\label{app:gradient-hessian}

Recall the expression in~\eqref{eq:dual-free-energy}:
\begin{equation}
F_{T}(\mathcal{Q},q)=\sup_{\mu\in\mathbb{R}^{c}}f(\mu),
\end{equation}
where
\begin{equation}
f(\mu)\coloneqq\mu\cdot q-T\ln Z_{T}(\mu).
\label{eq:legendre-dual-log-part-app-def}
\end{equation}
We are interested in determining the gradient and Hessian of the objective
function $f(\mu)$. To this end, let us adopt the shorthand $\mu\cdot
Q\equiv\sum_{i\in\left[  c\right]  }\mu_{i}Q_{i}$ and recall from Duhamel's formula and \cite[Eq.~(10) and Lemma~10]{patel2024quantumboltzmannmachinelearning}
that
\begin{align}
\frac{\partial}{\partial\mu_{i}}e^{-\frac{1}{T}\left(  H-\mu\cdot Q\right)  }
&  =\frac{1}{T}\int_{0}^{1}ds\ e^{-\frac{1}{T}\left(  H-\mu\cdot Q\right)
\left(  1-s\right)  }Q_{i}e^{-\frac{1}{T}\left(  H-\mu\cdot Q\right)
s}\label{eq:duhamel-deriv}\\
&  =\frac{1}{2T}\left\{  \Phi_{\mu}(Q_{i}),e^{-\frac{1}{T}\left(  H-\mu\cdot
Q\right)  }\right\}  ,
\end{align}
where $\Phi_{\mu}$ is the following quantum channel:
\begin{equation}
\Phi_{\mu}(\cdot)\coloneqq \int_{-\infty}^{\infty}dt\ p(t)\, e^{-i\left(  H-\mu\cdot
Q\right)  t/T}(\cdot)e^{i\left(  H-\mu\cdot Q\right)  t/T},
\end{equation}
and $p(t)$ is the following high-peak-tent probability density
\cite{patel2024quantumboltzmannmachinelearning}:
\begin{equation}
\label{eq:high-peak-tent-density}
p(t)\coloneqq \frac{\pi}{2}\ln\!\left\vert \cosh\!\left(  \frac{\pi t}{2}\right)
\right\vert .
\end{equation}

\begin{lemma}
For $i\in\left[  c\right]  $, the gradient of the function $f(\mu)$, as
defined in~\eqref{eq:legendre-dual-log-part-app-def}, has elements given by
\begin{equation}
\frac{\partial}{\partial\mu_{i}}f(\mu)=q_{i}-\left\langle Q_{i}\right\rangle
_{\rho_{T}(\mu)}.
\end{equation}

\end{lemma}

\begin{proof}
Consider that
\begin{align}
\frac{\partial}{\partial\mu_{i}}f(\mu)  &  =\frac{\partial}{\partial\mu_{i}
}\left(  \mu\cdot q-T\ln Z_{T}(\mu)\right) \\
&  =q_{i}-T\frac{\partial}{\partial\mu_{i}}\ln Z_{T}(\mu)\\
&  =q_{i}-T\frac{1}{Z_{T}(\mu)}\frac{\partial}{\partial\mu_{i}}Z_{T}(\mu).
\label{eq:proof-deriv-log-part}
\end{align}
Now consider that
\begin{align}
\frac{\partial}{\partial\mu_{i}}Z_{T}(\mu)  &  =\frac{\partial}{\partial
\mu_{i}}\operatorname{Tr}\!\left[  e^{-\frac{1}{T}\left(  H-\mu\cdot Q\right)
}\right] \\
&  =\frac{1}{2T}\operatorname{Tr}\!\left[  \left\{  \Phi_{\mu}(Q_{i}
),e^{-\frac{1}{T}\left(  H-\mu\cdot Q\right)  }\right\}  \right] \\
&  =\frac{1}{2T}\operatorname{Tr}\!\left[  \left\{  Q_{i},\Phi_{\mu}\!\left(
e^{-\frac{1}{T}\left(  H-\mu\cdot Q\right)  }\right)  \right\}  \right] \\
&  =\frac{1}{T}\operatorname{Tr}\!\left[  Q_{i}e^{-\frac{1}{T}\left(
H-\mu\cdot Q\right)  }\right]  .
\end{align}
Substituting into~\eqref{eq:proof-deriv-log-part}, then we find that
\begin{align}
q_{i}-T\frac{1}{Z_{T}(\mu)}\frac{\partial}{\partial\mu_{i}}Z_{T}(\mu)  &
=q_{i}-T\frac{1}{Z_{T}(\mu)}\frac{1}{T}\operatorname{Tr}\!\left[
Q_{i}e^{-\frac{1}{T}\left(  H-\mu\cdot Q\right)  }\right] \\
&  =q_{i}-\operatorname{Tr}[Q_{i}\rho_{T}(\mu)]\\
&  =q_{i}-\left\langle Q_{i}\right\rangle _{\rho_{T}(\mu)},
\end{align}
thus concluding the proof.
\end{proof}

\bigskip

The Kubo--Mori information matrix of the parameterized family $\left(
\rho_{T}(\mu)\right)  _{\mu\in\mathbb{R}^{c}}$ of thermal states has the
following form:
\begin{align}
I_{ij}^{\operatorname{KM}}(\mu)  &  =\frac{1}{T}\left(  \frac{1}
{2}\left\langle \left\{  \Phi_{\mu}(Q_{i}),Q_{j}\right\}  \right\rangle
_{\rho_{T}(\mu)}-\left\langle Q_{i}\right\rangle _{\rho_{T}(\mu)}\left\langle
Q_{j}\right\rangle _{\rho_{T}(\mu)}\right) \label{eq:kubo-mori-1}\\
&  =\frac{1}{T}\left(  \int_{0}^{1}ds\ \operatorname{Tr}[\rho_{T}(\mu
)^{1-s}Q_{i}\rho_{T}(\mu)^{s}Q_{j}]-\left\langle Q_{i}\right\rangle _{\rho
_{T}(\mu)}\left\langle Q_{j}\right\rangle _{\rho_{T}(\mu)}\right)  .
\label{eq:kubo-mori-2}
\end{align}

\begin{lemma}
\label{lem:hessian-kubo-mori}For $i,j\in\left[  c\right]  $, the Hessian of
the function $f(\mu)$, as defined in
\eqref{eq:legendre-dual-log-part-app-def}, has elements given by
\begin{equation}
\frac{\partial^{2}}{\partial\mu_{i}\partial\mu_{j}}f(\mu)=-I_{ij}
^{\operatorname{KM}}(\mu). \label{eq:hessian-neg-kubo-mori}
\end{equation}

\end{lemma}

\begin{proof}
Consider that
\begin{align}
\frac{\partial^{2}}{\partial\mu_{i}\partial\mu_{j}}f(\mu)  &  =\frac
{\partial^{2}}{\partial\mu_{i}\partial\mu_{j}}\left(  \mu\cdot q-T\ln
Z_{T}(\mu)\right) \\
&  =\frac{\partial}{\partial\mu_{i}}\left(  \frac{\partial}{\partial\mu_{j}
}\left(  \mu\cdot q-T\ln Z_{T}(\mu)\right)  \right) \\
&  =\frac{\partial}{\partial\mu_{i}}\left(  q_{j}-\operatorname{Tr}[Q_{j}
\rho_{T}(\mu)]\right) \\
&  =-\frac{\partial}{\partial\mu_{i}}\left(\frac{1}{Z_{T}(\mu)}\operatorname{Tr}
\!\left[  Q_{j}e^{-\frac{1}{T}\left(  H-\mu\cdot Q\right)  }\right]\right) \\
&  =\frac{1}{\left[  Z_{T}(\mu)\right]  ^{2}}\left(  \frac{\partial}
{\partial\mu_{i}}Z_{T}(\mu)\right)  \operatorname{Tr}\!\left[  Q_{j}
e^{-\frac{1}{T}\left(  H-\mu\cdot Q\right)  }\right] -\frac{1}{Z_{T}(\mu)}\operatorname{Tr}\!\left[  Q_{j}\frac{\partial
}{\partial\mu_{i}}e^{-\frac{1}{T}\left(  H-\mu\cdot Q\right)  }\right] \\
&  =\frac{1}{\left[  Z_{T}(\mu)\right]  ^{2}}\frac{1}{T}\operatorname{Tr}\!
\left[  Q_{i}e^{-\frac{1}{T}\left(  H-\mu\cdot Q\right)  }\right]
\operatorname{Tr}\!\left[  Q_{j}e^{-\frac{1}{T}\left(  H-\mu\cdot Q\right)
}\right] \nonumber\\
&  \qquad-\frac{1}{Z_{T}(\mu)}\operatorname{Tr}\!\left[  Q_{j}\frac{1}
{2T}\left\{  \Phi_{\mu}(Q_{i}),e^{-\frac{1}{T}\left(  H-\mu\cdot Q\right)
}\right\}  \right] \\
&  =\frac{1}{T}\left(  \left\langle Q_{i}\right\rangle _{\rho_{T}(\mu
)}\left\langle Q_{j}\right\rangle _{\rho_{T}(\mu)}-\frac{1}{2}
\operatorname{Tr}\!\left[  Q_{j}\left\{  \Phi_{\mu}(Q_{i}),\rho_{T}
(\mu)\right\}  \right]  \right) \\
&  =\frac{1}{T}\left(  \left\langle Q_{i}\right\rangle _{\rho_{T}(\mu
)}\left\langle Q_{j}\right\rangle _{\rho_{T}(\mu)}-\frac{1}{2}
\operatorname{Tr}\!\left[  \left\{  \Phi_{\mu}(Q_{i}),Q_{j}\right\}  \rho
_{T}(\mu)\right]  \right) \\
&  =\frac{1}{T}\left(  \left\langle Q_{i}\right\rangle _{\rho_{T}(\mu
)}\left\langle Q_{j}\right\rangle _{\rho_{T}(\mu)}-\frac{1}{2}\left\langle
\left\{  \Phi_{\mu}(Q_{i}),Q_{j}\right\}  \right\rangle _{\rho_{T}(\mu
)}\right)  ,
\end{align}
thus concluding the proof of~\eqref{eq:hessian-neg-kubo-mori} using the
expression in~\eqref{eq:kubo-mori-1}. To arrive at the other expression in
\eqref{eq:kubo-mori-2}, consider instead using the derivative expression in~\eqref{eq:duhamel-deriv}.
\end{proof}

\begin{lemma}
\label{lem:upper-bnd-matrix-elements-hessian}
The following upper bound holds for all $\mu\in\mathbb{R}^{c}$:
\begin{equation}
\left\vert \frac{\partial^{2}}{\partial\mu_{i}\partial\mu_{j}}f(\mu
)\right\vert \leq\frac{2}{T}\left\Vert Q_{i}\right\Vert \left\Vert
Q_{j}\right\Vert ,
\end{equation}
where $f(\mu)$ is defined in~\eqref{eq:legendre-dual-log-part-app-def}.
\end{lemma}

\begin{proof}
We give two different proofs. Consider from Lemma~\ref{lem:hessian-kubo-mori} that
\begin{align}
\left\vert \frac{\partial^{2}}{\partial\mu_{i}\partial\mu_{j}}f(\mu
)\right\vert  &  =\frac{1}{T}\left\vert \left\langle Q_{i}\right\rangle
_{\rho_{T}(\mu)}\left\langle Q_{j}\right\rangle _{\rho_{T}(\mu)}-\frac{1}
{2}\left\langle \left\{  \Phi_{\mu}(Q_{i}),Q_{j}\right\}  \right\rangle
_{\rho_{T}(\mu)}\right\vert \nonumber\\
&  \leq\frac{1}{T}\left\vert \left\langle Q_{i}\right\rangle _{\rho_{T}(\mu
)}\left\langle Q_{j}\right\rangle _{\rho_{T}(\mu)}\right\vert +\frac{1}
{2T}\left\vert \left\langle \left\{  \Phi_{\mu}(Q_{i}),Q_{j}\right\}
\right\rangle _{\rho_{T}(\mu)}\right\vert \\
&  =\frac{1}{T}\left\vert \left\langle Q_{i}\right\rangle _{\rho_{T}(\mu
)}\right\vert \left\vert \left\langle Q_{j}\right\rangle _{\rho_{T}(\mu
)}\right\vert +\frac{1}{2T}\left\vert \left\langle \left\{  \Phi_{\mu}
(Q_{i}),Q_{j}\right\}  \right\rangle _{\rho_{T}(\mu)}\right\vert \\
&  \leq\frac{1}{T}\left\Vert Q_{i}\right\Vert \left\Vert Q_{j}\right\Vert
+\frac{1}{2T}\left\Vert \left\{  \Phi_{\mu}(Q_{i}),Q_{j}\right\}  \right\Vert
\\
&  \leq\frac{1}{T}\left\Vert Q_{i}\right\Vert \left\Vert Q_{j}\right\Vert
+\frac{1}{2T}2\left\Vert Q_{i}\right\Vert \left\Vert Q_{j}\right\Vert \\
&  =\frac{2}{T}\left\Vert Q_{i}\right\Vert \left\Vert Q_{j}\right\Vert ,
\end{align}
where we used the fact that
\begin{align}
\left\Vert \left\{  \Phi_{\mu}(Q_{i}),Q_{j}\right\}  \right\Vert  &
\leq\left\Vert \Phi_{\mu}(Q_{i})Q_{j}\right\Vert +\left\Vert Q_{j}\Phi_{\mu
}(Q_{i})\right\Vert \\
&  \leq2\left\Vert \Phi_{\mu}(Q_{i})\right\Vert \left\Vert Q_{j}\right\Vert \\
&  \leq2\left\Vert Q_{i}\right\Vert \left\Vert Q_{j}\right\Vert ,
\end{align}
with the last inequality following from convexity and unitary invariance of
the spectral norm.

Alternatively, we could use Lemma~\ref{lem:hessian-kubo-mori} and the
expression in~\eqref{eq:kubo-mori-2} to arrive at this same bound. Consider
that
\begin{align}
\left\vert \frac{\partial^{2}}{\partial\mu_{i}\partial\mu_{j}}f(\mu
)\right\vert  &  =\frac{1}{T}\left\vert \left\langle Q_{i}\right\rangle
_{\rho_{T}(\mu)}\left\langle Q_{j}\right\rangle _{\rho_{T}(\mu)}-\int_{0}
^{1}ds\ \operatorname{Tr}[\rho_{T}(\mu)^{1-s}Q_{i}\rho_{T}(\mu)^{s}
Q_{j}]\right\vert \\
&  \leq\frac{1}{T}\left\vert \left\langle Q_{i}\right\rangle _{\rho_{T}(\mu
)}\left\langle Q_{j}\right\rangle _{\rho_{T}(\mu)}\right\vert +\frac{1}{T}
\int_{0}^{1}ds\ \left\vert \operatorname{Tr}[\rho_{T}(\mu)^{1-s}Q_{i}\rho
_{T}(\mu)^{s}Q_{j}]\right\vert \\
&  \leq\frac{1}{T}\left\vert \left\langle Q_{i}\right\rangle _{\rho_{T}(\mu
)}\right\vert \left\vert \left\langle Q_{j}\right\rangle _{\rho_{T}(\mu
)}\right\vert +\frac{1}{T}\int_{0}^{1}ds\ \left\Vert \rho_{T}(\mu)^{1-s}
Q_{i}\rho_{T}(\mu)^{s}Q_{j}\right\Vert _{1}\\
&  \leq\frac{1}{T}\left\Vert Q_{i}\right\Vert \left\Vert Q_{j}\right\Vert
+\frac{1}{T}\left\Vert \rho_{T}(\mu)^{1-s}\right\Vert _{1/\left(  1-s\right)
}\left\Vert Q_{i}\right\Vert \left\Vert \rho_{T}(\mu)^{s}\right\Vert
_{1/s}\left\Vert Q_{j}\right\Vert \\
&  =\frac{1}{T}\left\Vert Q_{i}\right\Vert \left\Vert Q_{j}\right\Vert
+\frac{1}{T}\left\Vert Q_{i}\right\Vert \left\Vert Q_{j}\right\Vert \\
&  =\frac{2}{T}\left\Vert Q_{i}\right\Vert \left\Vert Q_{j}\right\Vert .
\end{align}
The third inequality follows from a multivariate generalization of the
H\"{o}lder inequality (see, e.g.,~\cite[Eq.~(8)]{Beigi2013}).
\end{proof}

\begin{lemma}\label{lem:concavity}
The function $f(\mu)$, defined in~\eqref{eq:legendre-dual-log-part-app-def},
is concave in $\mu$.
\end{lemma}

\begin{proof}
This follows from Lemma~\ref{lem:hessian-kubo-mori} and because the Kubo--Mori
information matrix $I_{ij}^{\operatorname{KM}}(\mu)$ is known to be positive
semi-definite. See, e.g., \cite[Proposition~4]{minervini2025evolvedquantumboltzmannmachines}.

Here we give an alternative proof that $I^{\operatorname{KM}}(\mu)\geq0$. Let
$v\coloneqq\left(  v_{1},\ldots,v_{c}\right)  \in\mathbb{R}^{c}$. Defining
$W\coloneqq\sum_{i\in\left[  c\right]  }v_{i}Q_{i}$, $W_{0}\coloneqq W-\left\langle
W\right\rangle _{\rho_{T}(\mu)}I$, and observing that $W_{0}$ is Hermitian,
consider that
\begin{align}
v^{T}I^{\operatorname{KM}}(\mu)v  &  =\sum_{i,j=1}^{c}v_{i}I_{ij}
^{\operatorname{KM}}(\mu)v_{j}\\
&  =\frac{1}{T}\sum_{i,j=1}^{c}v_{i}v_{j}\left(  \int_{0}^{1}
ds\ \operatorname{Tr}[\rho_{T}(\mu)^{1-s}Q_{i}\rho_{T}(\mu)^{s}Q_{j}
]-\left\langle Q_{i}\right\rangle _{\rho_{T}(\mu)}\left\langle Q_{j}
\right\rangle _{\rho_{T}(\mu)}\right) \\
&  =\frac{1}{T}\left(  \int_{0}^{1}ds\ \operatorname{Tr}[\rho_{T}(\mu
)^{1-s}W\rho_{T}(\mu)^{s}W]-\left\langle W\right\rangle _{\rho_{T}(\mu
)}\left\langle W\right\rangle _{\rho_{T}(\mu)}\right) \\
&  =\frac{1}{T}\int_{0}^{1}ds\ \operatorname{Tr}\!\left[  \rho_{T}(\mu
)^{1-s}\left(  W-\left\langle W\right\rangle _{\rho_{T}(\mu)}I\right)
\rho_{T}(\mu)^{s}\left(  W-\left\langle W\right\rangle _{\rho_{T}(\mu
)}I\right)  \right] \\
&  =\frac{1}{T}\int_{0}^{1}ds\ \operatorname{Tr}\!\left[  \rho_{T}(\mu
)^{1-s}W_{0}\rho_{T}(\mu)^{s}W_{0}\right] \\
&  \geq0.
\end{align}
The final inequality follows because $\rho_{T}(\mu)^{1-s}$ and $\rho_{T}
(\mu)^{s}$ are positive semi-definite, $W_{0}$ is Hermitian, thus implying
that $W_{0}\rho_{T}(\mu)^{s}W_{0}$ is positive semi-definite, which in turn
implies that $\operatorname{Tr}\!\left[  \rho_{T}(\mu)^{1-s}W_{0}\rho_{T}
(\mu)^{s}W_{0}\right]  $ is non-negative for all $s\in\left[  0,1\right]  $.
\end{proof}

\begin{lemma}
\label{lem:hessian-upper-bound}
Let $F(\mu)$ be the Hessian matrix with matrix elements $\frac{\partial^{2}}
{\partial\mu_{i}\partial\mu_{j}}f(\mu)$, where $f(\mu)$ is defined in~\eqref{eq:legendre-dual-log-part-app-def}. Then the following bound holds:
\begin{equation}
\left\Vert F(\mu)\right\Vert \leq\frac{2}{T}\sum_{i\in\left[  c\right]
}\left\Vert Q_{i}\right\Vert ^{2}.
\end{equation}

\end{lemma}

\begin{proof}
Consider that $\left\Vert F(\mu)\right\Vert \leq\left\Vert F(\mu)\right\Vert
_{1}$; i.e., the spectral norm of $F(\mu)$ is bounded from above by its trace norm. Since
$F(\mu)$ is negative semi-definite (see Lemma~\ref{lem:hessian-kubo-mori} and proof of Lemma~\ref{lem:concavity}), it follows that
\begin{align}
\left\Vert F(\mu)\right\Vert _{1}  & =-\operatorname{Tr}[F(\mu)]\\
& =-\sum_{i\in\left[  c\right]  }\frac{\partial^{2}}{\partial\mu_{i}^{2}}
f(\mu)\\
& \leq\sum_{i\in\left[  c\right]  }\left\vert \frac{\partial^{2}}{\partial
\mu_{i}^{2}}f(\mu)\right\vert \\
& \leq\sum_{i\in\left[  c\right]  }\frac{2}{T}\left\Vert Q_{i}\right\Vert
^{2},
\end{align}
where the last inequality follows from Lemma~\ref{lem:upper-bnd-matrix-elements-hessian}.
\end{proof}

\section{Error analysis for Algorithm~\ref{alg:basic-grad-asc-free-energy}}

\label{app:alg-error-analysis}

As mentioned after Algorithm~\ref{alg:basic-grad-asc-free-energy}, it has
three sources of error:\ the error from approximating the minimum energy
$E(\mathcal{Q},q)$ by the minimum free energy $F_{T}(\mathcal{Q},q)$ at finite
temperature~$T$, the error from gradient ascent arriving at an approximate
global minimum of the free energy, and the error from outputting
\begin{equation}
\tilde{E} \equiv \langle H \rangle_{\rho_T(\mu^M)} + \mu^M\cdot
(q-\langle Q\rangle _{\rho_{T}(\mu^M)})    
\end{equation}
instead of
\begin{equation}
    \langle H \rangle_{\rho_T(\mu^M)} -TS(\rho_{T}
(\mu^{M})) + \mu^M\cdot
(q-\langle Q\rangle _{\rho_{T}(\mu^M)}) 
=  \mu^M\cdot q-T\ln
Z_{T}(\mu^M) = f(\mu^M),
\end{equation}
where the last equality follows from~\eqref{eq:log-part-back-to-free-energy-1}--\eqref{eq:log-part-back-to-free-energy-2}.
Let us label the first two errors
as $\varepsilon_{1}$ and $\varepsilon_{2}$, respectively, and we note that the
third error is related to the first, as seen below.

By setting $T=\frac{\varepsilon_{1}}{\ln d}$, it follows from
\eqref{eq:energy-free-energy-bounds} that
\begin{equation}
\left\vert F_{T}(\mathcal{Q},q)-E(\mathcal{Q},q)\right\vert \leq
\varepsilon_{1}.
\end{equation}
By picking the number of steps, $M$, of
Algorithm~\ref{alg:basic-grad-asc-free-energy} to satisfy $M\geq
\frac{L\left\Vert \mu^{\ast}\right\Vert^2 }{2\varepsilon_{2}}$, it follows from
\cite[Corollary~3.5]{garrigos2024handbookconvergencetheoremsstochastic} that
gradient ascent converges to an $\varepsilon_{2}$-approximate global maximum
after $M$ steps. This means that
\begin{equation}
\left\vert F_{T}(\mathcal{Q},q)-f(\mu^{M})   \right\vert \leq\varepsilon_{2}.
\end{equation}
By applying the triangle inequality, we then find that
\begin{align}
 \left\vert E(\mathcal{Q},q)-\tilde{E} \right\vert 
&  \leq\left\vert E(\mathcal{Q},q)-F_{T}(\mathcal{Q}
,q)\right\vert \notag   +\left\vert F_{T}(\mathcal{Q},q)- f(\mu^{M})  \right\vert
+\left\vert f(\mu^{M})  - \tilde{E} \right\vert \\
&  \leq\varepsilon_{1}+\varepsilon_{2}+TS(\rho_{T}(\mu^{M}))\\
&  \leq\varepsilon_{1}+\varepsilon_{2}+T\ln d\\
&  =2\varepsilon_{1}+\varepsilon_{2},
\end{align}
where we applied~\eqref{eq:von-neumann-entropy-ineqs}\ to arrive at the last
inequality. By picking $\varepsilon_{1}=\varepsilon/4$ and $\varepsilon
_{2}=\varepsilon/2$, we then conclude that $\left\vert E(\mathcal{Q}
,q)-\tilde{E} \right\vert \leq
\varepsilon$, as desired. This explains the choices for $T$ and $M$ made in
Algorithm~\ref{alg:basic-grad-asc-free-energy}.

\section{Proof of Lemma~\ref{lem:reduction-SDP-to-energy-min}}

\label{app:reduction-SDP-to-energy-min}Let us begin by adding a slack variable
$z\geq0$ to convert the inequality constraint $\operatorname{Tr}[X]\leq R$ in
\eqref{eq:R-SDP} to an equality constraint:
\begin{equation}
\alpha_{R}=\min_{X,z\geq0}\left\{  \operatorname{Tr}[CX]:\operatorname{Tr}
[X]+z=R,\ \operatorname{Tr}[A_{i}X]=b_{i}\ \forall i\in\left[  c\right]
\right\}  .
\end{equation}
We can then make the substitutions $W=X/R$, $y=z/R$, and $b_{i}^{\prime}
=b_{i}/R$ for all $i\in\left[  c\right]  $, leading to the follow rewrite of
the semi-definite program:
\begin{align}
\alpha_{R}  &  =\min_{X,z\geq0}\left\{
R\operatorname{Tr}[CX/R]:\operatorname{Tr}[X/R]+z/R=1,\ 
\operatorname{Tr}[A_{i}X/R]=b_{i}/R\ \forall i\in\left[  c\right]
\right\} \\
&  =R\min_{W,y\geq0}\left\{  \operatorname{Tr}[CW]:\operatorname{Tr}
[W]+y=1,\ \operatorname{Tr}[A_{i}W]=b_{i}^{\prime}\ \forall i\in\left[
c\right]  \right\}  .
\end{align}
Finally, let us define a quantum state $\rho$ as
\begin{equation}
\rho=
\begin{bmatrix}
W & v\\
v^{\dag} & y
\end{bmatrix}
,
\end{equation}
where $v\in\mathbb{C}^{d}$ is chosen such that $\rho\geq0$, and define the
observables $C^{\prime}\coloneqq C\oplus\left[  0\right]  $ and $A_{i}
^{\prime}\coloneqq A_{i}\oplus\left[  0\right]  $ for all $i\in\left[
c\right]  $. Then we find that
\begin{align}
1  &  =\operatorname{Tr}[\rho]=\operatorname{Tr}[W]+y,\\
\operatorname{Tr}[CW]  &  =\operatorname{Tr}[C^{\prime}\rho],\\
\operatorname{Tr}[A_{i}W]  &  =\operatorname{Tr}[A_{i}^{\prime}\rho],
\end{align}
and thus the desired equality in~\eqref{eq:reduction-to-states}\ holds.

\section{Stochastic gradient ascent analysis}

\subsection{Background on stochastic gradient ascent}

\label{sec:sgd}

In this appendix, we present some definitions and  results from the
optimization literature. Specifically, we revisit one of the known convergence
results~\cite[Theorem~6.3]{Bubeck2015} associated with the stochastic gradient ascent algorithm.

Consider the following maximization problem:
\begin{equation}
f^{\ast}\coloneqq\sup_{x\in\mathbb{R}^{c}}f(x),
\end{equation}
where $f\colon\mathbb{R}^{c}\rightarrow\mathbb{R}$ is an $\ell$-smooth
function and $f^{\ast}$ is the global maximum. Let $\mathcal{X}\subset
\mathbb{R}^{c}$ be such that $\mathcal{X}$ is a Euclidean ball of
radius $r\geq0$. If $x^{\ast}\coloneqq \operatorname{argmax}_{x\in\mathbb{R}^{c}}f(x)$
and $x^{\ast}\in\mathcal{X}$, then
\begin{equation}
f^{\ast}=\sup_{x\in\mathcal{X}}f(x).
\end{equation}
The projected stochastic gradient ascent algorithm (abbreviated SGA) updates
the iterate according to the following rule:
\begin{equation}
x^{m+1}=\Pi_{\mathcal{X}}(x^{m}+\eta\overline{g}(x^{m})),
\end{equation}
where the superscripts $m+1$ and $m$ on $x$ indicate the iteration of the algorithm, $\Pi_{\mathcal{X}}$ is the Euclidean projection onto $\mathcal{X}$
(i.e., $\Pi_{\mathcal{X}}(y)\coloneqq \operatorname{argmin}_{x\in\mathcal{X}
}\left\Vert y-x\right\Vert _{2}$), $\overline{g}(x)$ is a stochastic gradient
evaluated at some point $x$, and $\eta>0$ is the learning rate parameter.
Furthermore, the SGA algorithm requires the stochastic gradient $\overline
{g}(x)$ to be unbiased, i.e., $\mathbb{E}[\overline{g}(x)]=\nabla f(x)$, for
all $x\in\mathbb{R}^{c}$. Here, the expectation $\mathbb{E}[\cdot]$ is with
respect to the randomness inherent in generating $\overline{g}(x)$. In addition, for all
$x\in\mathbb{R}^{c}$, the stochastic gradient $\overline{g}(x)$ should also
satisfy the following bounded variance condition: there exists a constant
$\sigma\geq0$ such that
\begin{equation}
\mathbb{E}\!\left[  \left\Vert \overline{g}(x)-\nabla f(x)\right\Vert
^{2}\right]  \leq\sigma^{2}.
\label{eq:bounded-variance-condition}
\end{equation}

To this end, the following lemma demonstrates the rate at which SGA\ converges
to an $\varepsilon$-approximate global maximum of $f$. This lemma is a special
case of~\cite[Theorem~6.3]{Bubeck2015}, and we include it here for completeness.

\begin{lemma}
[SGA Convergence]
\label{lem:sgd_conv}
Let $f$ be concave and $L$-smooth on $\mathcal{X}$. Suppose
that $\mathcal{X}$\ is a Euclidean ball of radius $r\geq0$.
Suppose that we have access to a stochastic gradient oracle, which upon input
of $x\in\mathcal{X}$, returns a random vector $\overline{g}(x)$ satisfying
$\mathbb{E}\!\left[  \overline{g}(x)\right]  =\nabla f(x)$ and $\mathbb{E}\!
\left[  \left\Vert \overline{g}(x)-\nabla f(x)\right\Vert ^{2}\right]
\leq\sigma^{2}$. Then projected SGA\ with stepsize $\eta=\left[
L+1/\gamma\right]  ^{-1}$ and $\gamma=\frac{r}{\sigma}\sqrt{\frac{2}{M}}$
satisfies
\begin{equation}
f(x^{\ast})-\mathbb{E}\!\left[  f\!\left(  \overline{x^{M}}\right)  \right]  \leq
r\sigma\sqrt{\frac{2}{M}}+\frac{Lr^{2}}{M},
\end{equation}
where 
$
\overline{x^{M}}\coloneqq \frac{1}{M}\sum_{m=1}^{M}x^{m}$.

\end{lemma}

\subsection{Stochastic gradient estimation}

\label{app:grad_est} 

In this section, we present a quantum algorithm for estimating the gradient of the objective function $f(\mu)$ defined in~\eqref{eq:obj_fun}, whose elements are given by
\begin{equation}
\frac{\partial}{\partial\mu_{i}} f(\mu)=
\frac{\partial}{\partial\mu_{i}}\left(  \mu\cdot q-T\ln Z_{T}(\mu)\right)
= q_{i}-\left\langle Q_{i}\right\rangle _{\rho_{T}(\mu)}.
\label{eq:gradient-formula_app}
\end{equation}
The gradient is a multivariate vector-valued function with components corresponding to the partial derivatives in $\left(\frac{\partial}{\partial_{\mu_i}} f(\mu)\right)_i$. To estimate this gradient, Algorithm~\ref{alg:est_exp_Q} estimates each term $\left\langle Q_{i}\right\rangle _{\rho_{T}(\mu)}$ individually as a subroutine. The final gradient estimate is then constructed by combining these results, yielding the vector with components $q_i - \overline{Y}_i$, where $\overline{Y}_i$ is an estimate of $\left\langle Q_{i}\right\rangle _{\rho_{T}(\mu)}$. 
Assuming each $Q_i$ is expressed as a linear combination of Pauli strings as defined in~\eqref{eq:Q_pauli}, we can expand the expectation value as follows:
\begin{align}
   \left\langle Q_{i}\right\rangle _{\rho_{T}(\mu)} & = \operatorname{Tr}[Q_i \rho_T(\mu)]\\
& = \operatorname{Tr}\!\left[\left (\sum_{\vec{\jmath}} a_{i, \vec{j}} \sigma_{\vec{\jmath}} \right) \rho_T(\mu) \right]\\
& = \sum_{\vec{j}}a_{i, \vec{j}} \operatorname{Tr}\!\left [ \sigma_{\vec{\jmath}}\, \rho_T(\mu) \right].
\end{align}
Then, assuming that we have access to an oracle that samples an index $\vec{\jmath}$ with probability $\sfrac{a_{i,\vec{\jmath}}}{\left \Vert a_i \right \Vert_1}$ and that we have access to multiple copies of $\rho_T(\mu)$, Algorithm~\ref{alg:est_exp_Q} estimates the above quantity.

\begin{algorithm}[H]
\caption{$\mathtt{estimate\_obs}(\mu,   (a_{\vec\jmath})_{\vec\jmath}, \varepsilon, \delta)$}\label{alg:est_exp_Q} 

\begin{algorithmic}[1]
\STATE \textbf{Input:}
\begin{itemize}\setlength\itemsep{-0.15em}
    \item chemical potential vector $\mu \in \mathbb{R}^c$,
    \item non-negative coefficients $(a_{\vec\jmath})_{\vec\jmath}$ for the Pauli strings $(\sigma_{\vec{\jmath}})_{\vec\jmath}$,
    \item accuracy~$\varepsilon > 0$,
    \item error probability~$\delta \in (0,1)$
\end{itemize}
\STATE $N \leftarrow \sfrac{2\left \Vert a \right \Vert_1^2 \ln(\sfrac{2}{\delta})}{\varepsilon^2}$
\FOR{$n = 0$ to $N-1$}
\STATE Sample a multi-index $\vec{\jmath}$ with probability $\sfrac{a_{\vec{\jmath}}}{\left \Vert a \right \Vert_1}$
\STATE Prepare a register in the state $\rho_T(\mu)$, measure $\sigma_{\vec{\jmath}}$, and store the measurement outcome $b_n$
\STATE Set $Y_n \leftarrow \left \Vert a \right \Vert_1 (-1)^{b_n}$
\ENDFOR
\RETURN $\overline{Y} \leftarrow \frac{1}{N}\sum_{n=0}^{N-1}Y_n$
\end{algorithmic}
\end{algorithm}

Algorithm~\ref{alg:est_exp_Q} requires $N$ samples of the thermal state $\rho_T(\mu)$, with $N$ chosen to ensure convergence to the true expectation value within accuracy $\varepsilon$ and success probability $1-\delta$ (see Appendix~\ref{app:sample_complexity_gradient} for a derivation).

In one case of interest, for estimating $\left\langle Q_{i}\right\rangle _{\rho_{T}(\mu)}$, we call the subroutine $\mathtt{estimate\_obs}(\mu,   (a_{i,\vec\jmath})_{\vec\jmath}, \varepsilon, \delta)$, and suppose that $N_i$ is equal to the number of samples used. Algorithm~\ref{alg:est_exp_Q} outputs an unbiased estimator of the quantity of interest, i.e., $\mathbb{E}[\overline{Y}_i]= \left\langle Q_i \right\rangle _{\rho_T(\mu)}$, as we show below:
\begin{align}
    \mathbb{E}\!\left[\overline{Y}_i\right] & = \mathbb{E}\!\left[\frac{1}{N_i}\sum_{n=0}^{N_i-1}Y_n\right] = \mathbb{E}\!\left[\frac{1}{N_i}\sum_{n=0}^{N_i-1} \left \Vert a_i \right \Vert_1 (-1)^{b_n} \right]\\
    & = \frac{\left \Vert a_i \right \Vert_1}{N_i} \sum_{n=0}^{N_i-1} \mathbb{E}\!\left[  (-1)^{b_n} \right]  = \frac{\left \Vert a_i \right \Vert_1}{N_i} \sum_{n=0}^{N_i-1} \sum_{b_n \in\{0,1\}} p_{b_n}   (-1)^{b_n} ,\label{eq:unbias_proof}
\end{align}
where
\begin{equation}
    p_{b_n} \coloneqq \sum_{\vec{\jmath}} \frac{a_{i,\vec{\jmath}}}{\left \Vert a_i \right \Vert_1} \frac{1 + (-1)^{b_n}\operatorname{Tr}\! \left[ \sigma_{\vec{\jmath}} \rho_T(\mu) \right]} {2}.
\end{equation}
The above expression for $p_{b_n}$ follows from the fact that we first sample a multi-index $\vec{\jmath}$ with probability $\sfrac{a_{i,\vec{\jmath}}}{\left \Vert a_i \right \Vert_1}$ and then we measure the Pauli string $\sigma_{\vec{\jmath}}$ on the state $\rho_T(\mu)$.
Plugging the expression above for $p_{b_n}$ into~\eqref{eq:unbias_proof}, we finally conclude that
\begin{align}
    \mathbb{E}\!\left[\overline{Y}_i\right] & = \sum_{\vec{\jmath}} a_{i,\vec{\jmath}} \operatorname{Tr}\! \left[ \sigma_{\vec{\jmath}}\, \rho_T(\mu) \right] \\
    & = \sum_{\vec{\jmath}}  \operatorname{Tr}\! \left[ a_{i,\vec{\jmath}}\, \sigma_{\vec{\jmath}}\,  \rho_T(\mu)  \right] \\
    & =  \operatorname{Tr}\! \left[ Q_i  \rho_T(\mu)  \right] \\
    & = \left\langle Q_i \right\rangle _{\rho_T(\mu)}.
    \label{eq:expectation-qi-sampavg}
\end{align}

\subsection{Sample complexity}

\label{sec:sample_comp}

In this section, we investigate the sample complexity -- the number of samples of the thermal state~$\rho_T(\mu)$ -- required by Algorithm~\ref{alg:energy-min}  to reach a point $\varepsilon$-close to the optimum. To simplify the discussion, we divide the analysis into two parts. First, we investigate the sample complexities of the gradient estimation in Algorithm~\ref{alg:est_exp_Q} and then analyze the total sample complexity of the energy minimization algorithm itself. This is because the overall energy minimization algorithm employs Algorithm~\ref{alg:est_exp_Q} as a subroutine for stochastic gradient estimation.

\subsubsection{Sample complexity of stochastic gradient estimation}

\label{app:sample_complexity_gradient}

Recall that Algorithm~\ref{alg:est_exp_Q} outputs estimates of the second term appearing in the partial derivative $\frac{\partial_{i}}{\partial \mu_i} f(\mu)$ in~\eqref{eq:gradient-formula_app}. We demonstrated in Appendix~\ref{app:grad_est} that this estimator is unbiased. The next step is to investigate how fast this estimator converges to its expected value, which we do formally in the proof of the following lemma.

\begin{lemma}[Sample complexity of Algorithm~\ref{alg:est_exp_Q}]\label{lem:exp_Q}
    Let $\varepsilon > 0$, $\delta \in (0, 1)$, and Pauli-coefficient vector $a_i$. Then, the number of samples, $N_i$, of $\rho_T(\mu)$ used by Algorithm~\ref{alg:est_exp_Q} to produce an $\varepsilon$-close estimate of $\left\langle Q_i \right\rangle _{\rho_T(\mu)}$ with a success probability not less than $1 - \delta$ is
\begin{equation}
    N_i = \left \lceil \frac{2\left \Vert a_i \right \Vert_1^2 \ln(\sfrac{2}{\delta})}{\varepsilon^2} \right\rceil.
\end{equation}
\end{lemma}

\begin{proof}
     Recall that Algorithm~\ref{alg:est_exp_Q} outputs an unbiased estimator $\overline{Y}$ of $\left\langle Q_i \right\rangle _{\rho_T(\mu)}$: 
\begin{equation}
    \overline{Y} = \frac{1}{N}\sum_{n=0}^{N-1}Y_n,\label{eq:Z2}
\end{equation} 
where $Y_n = \left \Vert a_i \right \Vert_1 (-1)^{b_n}$ and $b_n \in \{0, 1\}$, for all $n \in \{0, \ldots, N-1\}$. This implies that $Y_n$ lies in the range
\begin{equation}
   -\left \Vert a_i \right \Vert_1 \leq Y_n \leq \left \Vert a_i \right \Vert_1. 
\end{equation}
Now, using the Hoeffding inequality (recalled as Lemma~\ref{lem:hoeffding} below), we can say that for $\varepsilon > 0$ and $\delta \in (0,1)$, we have
\begin{equation}
    \Pr\!\left( \left |\overline{Y} - \mathbb{E}\!\left[ \overline{Y}\right] \right|  \geq \varepsilon \right) \leq \delta
\end{equation}
if
\begin{equation}
    N_i \geq \frac{2\left \Vert a_i \right \Vert_1^2 \ln(\sfrac{2}{\delta})}{\varepsilon^2}.
\end{equation}
This is equivalent to 
\begin{equation}
    \Pr \! \left (\left |\overline{Y} - \mathbb{E}\!\left[ \overline{Y}\right] \right| \leq \varepsilon \right) \geq 
    1 - \delta ,\label{eq:Z(2)_bound}
\end{equation}
thus concluding the proof after recalling~\eqref{eq:expectation-qi-sampavg}.
\end{proof}

\begin{lemma}[Hoeffding's Bound~\cite{Hoeffding1963}]
\label{lem:hoeffding}
Suppose that we are given $n$ independent samples $Y_1, \ldots, Y_n$ of a bounded random variable $Y$ taking values in the interval $[a,b]$ and having mean $\mu$. Set $
    \overline{Y_n} \coloneqq \frac{1}{n} (Y_1 + \cdots +Y_n)$
to be the sample mean. Let $\varepsilon > 0 $ be the desired accuracy, and let $1-\delta$ be the desired success probability, where $\delta \in (0,1)$. Then
\begin{equation}
\label{eq:Hoeff-bound}
\Pr[\vert \overline{Y_n} - \mu \vert \leq \varepsilon] \geq 1-\delta,
\end{equation}
if
$    n \geq \frac{M^2}{2\varepsilon^2} \ln \!\left( \frac{2}{\delta}\right)$,
where $M \coloneqq b -  a$.
\end{lemma}

\subsubsection{Sample complexity of stochastic gradient ascent}

\label{app:sample-comp-SGD}

Using the development above, in the proof of the following theorem, we analyze
the sample complexity of the energy minimization algorithm
(Algorithm~\ref{alg:energy-min}). The following is a restatement of Theorem~\ref{thm:main} in the main text, with the sample complexity stated precisely rather than with big-$O$ notation.

\begin{theorem}
[Sample complexity of energy minimization]Let $H$ and $Q_{i}$ for each
$i\in\lbrack c]$, $c\in\mathbb{N}$, be a Hamiltonian and observables as
defined in~\eqref{eq:H_pauli} and~\eqref{eq:Q_pauli} respectively, and let $h$ and 
$a_{i}$ be the Pauli-coefficient vectors of $H$ and $Q_{i}$, respectively. Let $\varepsilon>0$. Then
to reach, in expectation and with success probability $\geq1-\delta$, an
$\varepsilon$-approximate estimate $\hat{E}$ of the optimal value
in~\eqref{eq:energy-minimization}, i.e., such that
\begin{equation}
\Pr\!\left[\left\vert E(\mathcal{Q},q)-\mathbb{E}\!\left[\hat{E}\right]\right\vert \leq\varepsilon\right] \geq 1-\delta ,
\label{eq:SGA-guarantee}
\end{equation}
the sample complexity of Algorithm~\ref{alg:energy-min} is given by
\begin{equation}
\left\lceil \frac{16r^2}{\varepsilon^{2}}\left(  2  c\varepsilon
^{2}
+\left(8  \ln d +2\delta\right) \sum_{i\in\left[  c\right]  }\left\Vert
a_{i}\right\Vert _{1}^2\right)  \right\rceil \sum_{i\in\left[  c\right]  }\left\lceil
\frac{2\left\Vert a_{i}\right\Vert _{1}^{2}\ln\!\left(  \frac{2}{\delta}\right)
}{\varepsilon^{2}}\right\rceil 
+\left\lceil \frac{8  }{\varepsilon^{2}}\left(  \left\Vert h\right\Vert _{1}+r\sqrt{\sum_{i\in\left[  c\right]  }\left\Vert a_{i}\right\Vert _{1}^{2}}\right)^2  \ln\!\left(  \frac{2}{\delta
}\right)\right\rceil .
\end{equation}
In~\eqref{eq:SGA-guarantee}, the expectation $\mathbb{E}$ is with respect to the $M$ steps of SGA updates, and the probability $\Pr$ is with respect to the final step in estimating $ \left\langle H\right\rangle _{\rho
_{T}(\overline{\mu^{M}})}+\overline{\mu^{M}}\cdot\left(  q-\left\langle
Q\right\rangle _{\rho_{T}(\overline{\mu^{M}})}\right)$. Additionally, the radius $r$ is such that $r\geq \left \| \mu^*\right\|$, where $\mu^*$ satisfies $f(\mu^{\ast})   =F_{T}(\mathcal{Q},q)$ for $T = \frac{\varepsilon}{4\ln d}$ and $f$ defined in~\eqref{eq:obj_fun}.
\end{theorem}

\begin{proof}
Note that Algorithm~\ref{alg:energy-min} is a stochastic gradient ascent
algorithm, where the stochastic gradients $\overline{g}(\mu)$, at any given
point $\mu$, are estimated using Algorithm~\ref{alg:est_exp_Q}:
\begin{equation}
\overline{g}(\mu)=\left(  \overline{g}_{1}(\mu),\ldots,\overline{g}_{c}
(\mu)\right)  ^{\mathsf{T}},
\end{equation}
where $\overline{g}_{i}(\mu)$ is the stochastic partial derivative, given as
\begin{equation}
\overline{g}_{i}(\mu)=q_{i}-\overline{Y}_{i}(\mu).
\label{eq:stoch_par_deri}
\end{equation}
Here, Algorithm~\ref{alg:est_exp_Q} evaluates $\overline{Y}_{i}(\mu)$.

From Appendix~\ref{sec:sgd}, we know that in order to use SGA for
optimization, the stochastic gradient should be unbiased. This is true for our
case; i.e., for all $\mu\in\mathbb{R}^{c}$, we have $\mathbb{E}[\overline
{g}(\mu)]=\nabla_{\mu}\left(  \mu\cdot q-T\ln Z_{T}(\mu)\right)  $, and we
showed this previously in Section~\ref{app:grad_est}.

Another requirement for SGA is that the variance of the stochastic gradient
should be bounded from above. Specifically, the stochastic gradient should
satisfy the condition given by~\eqref{eq:bounded-variance-condition} for some constant $\sigma$. Therefore, we
now proceed to obtain this constant for our case. Consider that
\begin{align}
\mathbb{E}\!\left[  \left\Vert \overline{g}(\mu)-\nabla_{\mu}f(\mu)\right\Vert
^{2}\right]   &  =\mathbb{E}\!\left[  \sum_{i\in\left[  c\right]  }\left(  \overline{g}
_{i}(\mu)-\partial_{i}f(\mu)\right)  ^{2}\right]  \\
&  =\sum_{i\in\left[  c\right]  }\mathbb{E}\!\left[  \left(  \overline{g}_{i}(\mu
)-\partial_{i}f(\mu)\right)  ^{2}\right]  \\
&  =\sum_{i\in\left[  c\right]  }\mathbb{E}\!\left[  \left(  q_{i}-\overline{Y}_{i}
(\mu)-\left[  q_{i}-\left\langle Q_{i}\right\rangle _{\rho_{T}(\mu)}\right]
\right)  ^{2}\right]  \\
&  =\sum_{i\in\left[  c\right]  }\operatorname{Var}\!\left[  \overline{Y}_{i}(\mu)\right]  \\
&  \leq\sum_{i\in\left[  c\right]  }(\varepsilon_{c}^{2}+\delta_{c}\left\Vert a_{i}
\right\Vert _{1}^{2})\\
&  \leq c\varepsilon_{c}^{2}+\delta_{c}\sum_{i\in\left[  c\right]  }\left\Vert a_{i}
\right\Vert _{1}^{2}.
\end{align}
The third equality follows from~\eqref{eq:stoch_par_deri}
and~\eqref{eq:gradient-formula_app}. The first inequality
follows directly from the variance bounds of the sample mean $\overline{Y}
_{i}(\mu)$. Indeed, this is a consequence of the following reasoning. Letting
$Y\in\lbrack-J,J]$ be a random variable such that $\Pr(|Y-\mathbb{E}
[Y]|\leq\varepsilon^{\prime})\geq1-\delta^{\prime}$, for $\varepsilon^{\prime
}>0$ and $\delta^{\prime}\in(0,1)$, and defining the set $\mathcal{S}
\coloneqq\{y:|y-\mathbb{E}[Y]|\leq\varepsilon^{\prime}\}$, we find that
\begin{align}
\operatorname{Var}\!\left[  Y\right]   &  =\sum_{y}p(y)|y-\mathbb{E}[Y]|^{2}\\
&  =\sum_{y\in\mathcal{S}}p(y)|y-\mathbb{E}[Y]|^{2}+\sum_{y\in\mathcal{S}^{c}
}p(y)|y-\mathbb{E}[Y]|^{2}\\
&  \leq\sum_{y\in\mathcal{S}}p(y)\varepsilon^{\prime2}+\sum_{y\in
\mathcal{S}^{c}}p(y)J^{2}\\
&  \leq\varepsilon^{\prime2}+\delta^{\prime2}J^{2}.
\end{align}
Applying this inequality, we
conclude the first inequality. Now, comparing the second inequality with the
condition given by~\eqref{eq:bounded-variance-condition}, we obtain the constant $\sigma$ for our
case:
\begin{equation}
\sigma^{2}=c\varepsilon_{c}^{2}+\delta_{c}\sum_{i\in\left[  c\right]  }\left\Vert
a_{i}\right\Vert _{1}^{2}.
\label{eq:sigma-variance-SGA}
\end{equation}

Recall from the convergence result of SGA (Lemma~\ref{lem:sgd_conv}) that, to
obtain an error $\varepsilon_{2}>0$, the total number of iterations, $M$,
should satisfy
\begin{equation}
r\sigma\sqrt{\frac{2}{M}}+\frac{Lr^{2}}{M}\leq\varepsilon_{2}.
\end{equation}
By simple algebra, this inequality can be rewritten as
\begin{equation}
M-\frac{r\sigma\sqrt{2}}{\varepsilon_{2}}\sqrt{M}-\frac{Lr^{2}}{\varepsilon
_{2}}\geq0,
\end{equation}
Applying the quadratic formula and solving for $\sqrt{M}$, we find that the above inequality
is satisfied if the following one holds:
\begin{equation}
M\geq\frac{1}{4}\left[  \frac{r\sigma\sqrt{2}}{\varepsilon_{2}}+\sqrt{\frac{2r^{2}
\sigma^{2}}{\varepsilon_{2}^{2}}+\frac{4Lr^2}{\varepsilon_{2}}}\right]
^{2}.\label{eq:SGA-num-steps-bound}
\end{equation}
With this choice, it follows from Lemma~\ref{lem:sgd_conv} that
\begin{equation}
\left\vert f(\mu^{\ast})-\mathbb{E}\!\left[  f\!\left(  \overline{\mu^{M}}\right)
\right]  \right\vert \leq\varepsilon_{2},
\label{eq:eps2-bnd-proof-SGA-convergence}
\end{equation}
for $f(\mu^{\ast})$ as given in~\eqref{eq:legendre-dual-log-part-app-def} and where
\begin{equation}
\overline{\mu^{M}}\coloneqq \frac{1}{M}\sum_{m=1}^{M}\mu^{m}.
\end{equation}
Additionally, the expectation $\mathbb{E}$ in~\eqref{eq:eps2-bnd-proof-SGA-convergence} is with respect to the randomness generated in the $M$ gradient update steps of the SGA algorithm.
Now recalling from~\eqref{eq:log-part-back-to-free-energy-1}--\eqref{eq:log-part-back-to-free-energy-2} that
\begin{align}
f(\mu^{\ast})  & =F_{T}(\mathcal{Q},q),\\
\mathbb{E}\!\left[  f\!\left(  \overline{\mu^{M}}\right)  \right]    &
=\mathbb{E}\!\left[  \left\langle H\right\rangle _{\rho_{T}(\overline{\mu^{M}}
)}-TS(\rho_{T}(\overline{\mu^{M}}))+\overline{\mu^{M}}\cdot\left(
q-\left\langle Q\right\rangle _{\rho_{T}(\overline{\mu^{M}})}\right)  \right]
,
\end{align}
it follows that
\begin{equation}
\left\vert F_{T}(\mathcal{Q},q)-\mathbb{E}\!\left[  f\!\left(  \overline{\mu^{M}
}\right)  \right]  \right\vert \leq\varepsilon_{2}.
\end{equation}
Now set
\begin{align}
\tilde{E} & \equiv\mathbb{E}\!\left[  \left\langle H\right\rangle _{\rho
_{T}(\overline{\mu^{M}})}+\overline{\mu^{M}}\cdot\left(  q-\left\langle
Q\right\rangle _{\rho_{T}(\overline{\mu^{M}})}\right)  \right]  \\
& = \mathbb{E}\!\left[ \overline{\mu^{M}}\cdot  q + \left\langle H -\overline{\mu^{M}}\cdot
Q\right\rangle _{\rho_{T}(\overline{\mu^{M}})}  \right],
\end{align}
and let $\hat{E}$ be the estimate of $\tilde{E}$ when using Algorithm~\ref{alg:energy-min}. That is, $\tilde{E}$ is the estimate that results when feeding in the observable $H -\overline{\mu^{M}}\cdot
Q$ to Algorithm~\ref{alg:est_exp_Q}, adding $\overline{\mu^{M}}\cdot  q$ to  the result, and then taking the expectation $\mathbb{E}$.
Applying an analysis similar to that in Appendix~\ref{app:alg-error-analysis}, and setting
$T=\frac{\varepsilon_{1}}{\ln d}$, we find that
\begin{align}
\left\vert E(\mathcal{Q},q)-\hat{E}\right\vert  & \leq\left\vert
E(\mathcal{Q},q)-F_{T}(\mathcal{Q},q)\right\vert +\left\vert F_{T}
(\mathcal{Q},q)-\mathbb{E}\!\left[  f\!\left(  \overline{\mu^{M}}\right)  \right]
\right\vert \nonumber\\
& \qquad+\left\vert \mathbb{E}\!\left[  f\!\left(  \overline{\mu^{M}}\right)
\right]  -\tilde{E}\right\vert +\left\vert \tilde{E}-\hat{E}\right\vert \\
& \leq\varepsilon_{1}+\varepsilon_{2}+T\mathbb{E}\!\left[  S(\rho_{T}
(\overline{\mu^{M}}))\right]  +\varepsilon_{3}\\
& \leq\varepsilon_{1}+\varepsilon_{2}+T\ln d+\varepsilon_{3}\\
& =2\varepsilon_{1}+\varepsilon_{2}+\varepsilon_{3}.
\end{align}

Thus, to have an overall error of $\varepsilon$, we can pick $\varepsilon
_{1}=\varepsilon_{2}=\varepsilon_{3}=\frac{\varepsilon}{4}$. We should then
set $T=\frac{\varepsilon}{4\ln d}$, and looking back to
\eqref{eq:SGA-num-steps-bound}, we require
\begin{align}
M  & \geq\frac{1}{4}\left[  \frac{r\sigma\sqrt{2}}{\varepsilon_{2}}+\sqrt{\frac
{2r^{2}\sigma^{2}}{\varepsilon_{2}^{2}}+\frac{4Lr^2}{\varepsilon_{2}}}\right]
^{2}\label{eq:M-ineq-sC}\\
& =\frac{1}{4}\left[  \frac{4\sqrt{2}r\sigma}{\varepsilon}+\sqrt{\frac{32r^{2}\sigma^{2}
}{\varepsilon^{2}}+\frac{16Lr^2}{\varepsilon}}\right]  ^{2}\\
& =4\left[  \frac{\sqrt{2}r\sigma}{\varepsilon}+\sqrt{\frac{2r^{2}\sigma^{2}
}{\varepsilon^{2}}+\frac{Lr^2}{\varepsilon}}\right]  ^{2}\\
& =4\left[  \frac{\sqrt{2}r\sigma}{\varepsilon}+\sqrt{\frac{2r^{2}\sigma^{2}
}{\varepsilon^{2}}+\frac{8r^2\left(  \ln d\right)  \sum_{i\in\left[  c\right]
}\left\Vert a_{i}\right\Vert _{1}^2}{\varepsilon^{2}}}\right]  ^{2}\\
& =\frac{4r^2}{\varepsilon^{2}}\left[  \sqrt{2}\sigma+\sqrt{2\sigma
^{2}+8\left(  \ln d\right)  \sum_{i\in\left[  c\right]  }\left\Vert
a_{i}\right\Vert _{1}^2}\right]  ^{2}.
\end{align}
where we plugged in the choice of the smoothness parameter $L$ from~\eqref{eq:smoothness-parameter-choice}, i.e.,
\begin{equation}
L=\frac{2}{T}\sum_{i\in\left[  c\right]  }\left\Vert a_{i}\right\Vert
_{1}^2=\frac{8\ln d}{\varepsilon}\sum_{i\in\left[  c\right]  }\left\Vert
a_{i}\right\Vert _{1}^2\geq\frac{2}{T}\sum_{i\in\left[  c\right]  }\left\Vert
Q_{i}\right\Vert^2\label{eq:alternate-L} .
\end{equation}
We can simplify the lower bound above by noticing that
\begin{align}
 \frac{4r^2}{\varepsilon^{2}}\left[  \sqrt{2}\sigma+\sqrt{2\sigma
^{2}+8\left(  \ln d\right)  \sum_{i\in\left[  c\right]  }\left\Vert
a_{i}\right\Vert _{1}^2}\right]  ^{2}
& =\frac{4r^2}{\varepsilon^{2}}\left[  \sqrt{2\sigma^{2}}+\sqrt
{2\sigma^{2}+8\left(  \ln d\right)  \sum_{i\in\left[  c\right]
}\left\Vert a_{i}\right\Vert _{1}^2}\right]  ^{2}\\
& \leq\frac{4r^2}{\varepsilon^{2}}\left[  2\sqrt{2\sigma^{2}+8\left(  \ln
d\right)  \sum_{i\in\left[  c\right]  }\left\Vert a_{i}\right\Vert _{1}^2
}\right]  ^{2}\\
& =\frac{16r^2}{\varepsilon^{2}}\left(  2\sigma^{2}+8\left(  \ln d\right)
\sum_{i\in\left[  c\right]  }\left\Vert a_{i}\right\Vert _{1}^2\right)  .
\end{align}
So we set $M$ as follows for simplicity:
\begin{equation}
M=\left\lceil \frac{16r^2}{\varepsilon^{2}}\left(  2\sigma^{2}+8\left(  \ln
d\right)  \sum_{i\in\left[  c\right]  }\left\Vert a_{i}\right\Vert
_{1}^2\right)  \right\rceil ,
\end{equation}
and with this choice, we are guaranteed that~\eqref{eq:M-ineq-sC} holds.
This resolves the minimum number of iterations used by SGA to reach an
$\varepsilon$-approximate optimum point.

Similarly, we can evaluate the step size $\eta$ for SGA, which we do in
the following way. Again recall from Lemma~\ref{lem:sgd_conv} that the step
size $\eta$ is given as follows:
\begin{equation}
\eta=\left[  L+1/\gamma\right]  ^{-1}=\left[  \frac{8\ln d}{\varepsilon}
\sum_{i\in\left[  c\right]  }\left\Vert a_{i}\right\Vert _{1}^2+\frac{\sigma}
{r}\sqrt{\frac{M}{2}}\right]  ^{-1}.\label{eq:eta_exp}
\end{equation}
Recall that each sample complexity $N_{i}$ of Algorithm~\ref{alg:est_exp_Q} is
given by
\begin{equation}
N_{i}=\left\lceil \frac{2\left\Vert a_{i}\right\Vert _{1}^{2}\ln\!\left(
\frac{2}{\delta_{c}}\right)  }{\varepsilon_{c}^{2}}\right\rceil .
\end{equation}
From this, we get the total sample complexity for gradient updates in 
Algorithm~\ref{alg:energy-min}:
\begin{align}
N  & =M\cdot\sum_{i\in\left[  c\right]  }N_{i}\\
& =\left\lceil \frac{16r^2}{\varepsilon^{2}}\left(  2\sigma^{2}+8\left(
\ln d\right)  \sum_{i\in\left[  c\right]  }\left\Vert a_{i}\right\Vert
_{1}^2\right)  \right\rceil \sum_{i\in\left[  c\right]  }\left\lceil \frac{2\left\Vert
a_{i}\right\Vert _{1}^{2}\ln\!\left(  \frac{2}{\delta_{c}}\right)  }
{\varepsilon_{c}^{2}}\right\rceil .
\end{align}
For simplicity, we can set $\varepsilon_{c}=\varepsilon$, $\delta
_{c}=\delta$, substitute \eqref{eq:sigma-variance-SGA}, and the above reduces to
\begin{align}
N & =\left\lceil \frac{16r^2}{\varepsilon^{2}}\left(  2\left(  c\varepsilon
^{2}+\delta\sum_{i\in\left[  c\right]  }\left\Vert a_{i}\right\Vert _{1}^{2}\right)
+8\left(  \ln d\right)  \sum_{i\in\left[  c\right]  }\left\Vert
a_{i}\right\Vert _{1}^2\right)  \right\rceil \sum_{i\in\left[  c\right]  }\left\lceil
\frac{2\left\Vert a_{i}\right\Vert _{1}^{2}\ln\!\left(  \frac{2}{\delta}\right)
}{\varepsilon^{2}}\right\rceil \\
& =\left\lceil \frac{16r^2}{\varepsilon^{2}}\left(  2  c\varepsilon
^{2}
+\left(8  \ln d +2\delta\right) \sum_{i\in\left[  c\right]  }\left\Vert
a_{i}\right\Vert _{1}^2\right)  \right\rceil \sum_{i\in\left[  c\right]  }\left\lceil
\frac{2\left\Vert a_{i}\right\Vert _{1}^{2}\ln\!\left(  \frac{2}{\delta}\right)
}{\varepsilon^{2}}\right\rceil,
\end{align}
such that the dependence of $N$ on $\varepsilon$ is $\varepsilon^{-4}$.

The
last step of Algorithm~\ref{alg:energy-min} involves an extra estimation of $\hat{E}$. By
invoking the Hoeffding bound (Lemma~\ref{lem:hoeffding}) and noticing that
\begin{align}
\left\vert \left\langle H-\mu\cdot Q\right\rangle _{\rho_{T}(\mu)}\right\vert& = 
\left\vert \left\langle H\right\rangle _{\rho_{T}(\mu)}-\mu\cdot\left\langle Q\right\rangle _{\rho_{T}(\mu)}\right\vert \\& \leq
\left\vert \left\langle H\right\rangle _{\rho_{T}(\mu)}\right| + \left|\mu\cdot \left\langle Q\right\rangle _{\rho_{T}(\mu)}\right|\\
& \leq \left \| H\right\| + \left \| \mu \right\| \sqrt{ \sum_{i\in\left[  c\right]  } \left|\left\langle Q_i\right\rangle _{\rho_{T}(\mu)}\right|^2}\\
& \leq \left \| H\right\| + r \sqrt{ \sum_{i\in\left[  c\right]  } \left\| Q_i\right\| ^2}\\
& \leq \left\Vert h\right\Vert _{1}+r \sqrt{\sum_{i\in\left[  c\right]  }\left\Vert a_{i}\right\Vert _{1}^{2}}
\end{align}
for all $\mu$ in a Euclidean ball of radius $r$,
we conclude that the last step of Algorithm~\ref{alg:energy-min} has sample complexity
\begin{equation}
\left\lceil \frac{8  }{\varepsilon^{2}}\left(  \left\Vert h\right\Vert _{1}+r\sqrt{\sum_{i\in\left[  c\right]  }\left\Vert a_{i}\right\Vert _{1}^{2}}\right)^2  \ln\!\left(  \frac{2}{\delta
}\right)\right\rceil ,
\end{equation}
in order to achieve $\varepsilon_{3}=\varepsilon/4$ error with probability
$1-\delta$.\ This concludes the proof.
\end{proof}

\section{Proof of Lemma~\ref{lem:reduction-SDP-to-energy-min_quantum}}

\label{app:reduction-SDP-to-energy-min_quantum}

Let us begin by adding a slack variable $z\geq0$ to convert the inequality
constraint $\operatorname{Tr}[X]\leq R$ in \eqref{eq:R-SDP} to an equality
constraint:
\begin{equation}
\alpha_{R}=\max_{X,z\geq0}\left\{  \operatorname{Tr}[CX]:\operatorname{Tr}
[X]+z=R,\ \operatorname{Tr}[A_{i}X]=b_{i}\ \forall i\in\left[  m\right]
\right\}  .
\end{equation}
We can then make the substitutions $W=X/R$, $y=z/R$, and $b_{i}^{\prime}
=b_{i}/R$ for all $i\in\left[  m\right]  $, leading to the follow rewrite of
the semi-definite program:
\begin{align}
\alpha_{R} &  =\max_{X,z\geq0}\left\{
R\operatorname{Tr}[CX/R]:\operatorname{Tr}[X/R]+z/R=1,\
\operatorname{Tr}[A_{i}X/R]=b_{i}/R\ \forall i\in\left[  m\right]
\right\}  \\
&  =R\max_{W,y\geq0}\left\{  \operatorname{Tr}[CW]:\operatorname{Tr}
[W]+y=1,\ \operatorname{Tr}[A_{i}W]=b_{i}^{\prime}\ \forall i\in\left[
m\right]  \right\}  .
\end{align}
Finally, let us define a quantum state $\rho\in\mathbb{D}_{2d}$ as
\begin{align}
\rho &  =
\begin{bmatrix}
W & W_{01}\\
W_{01}^{\dag} & W_{11}
\end{bmatrix}
\\
&  =|0\rangle\!\langle0|\otimes W+|0\rangle\!\langle1|\otimes W_{01}
+|1\rangle\langle0|\otimes W_{01}^{\dag}+|1\rangle\!\langle1|\otimes W_{11},
\end{align}
where $y=\operatorname{Tr}[W_{11}]$ and define the observables $C^{\prime
}\coloneqq|0\rangle\!\langle0|\otimes C$ and $A_{i}^{\prime}\coloneqq|0\rangle
\!\langle0|\otimes A_{i}$ for all $i\in\left[  m\right]  $. Then we find that
\begin{align}
1 &  =\operatorname{Tr}[\rho]=\operatorname{Tr}[W]+y,\\
\operatorname{Tr}[CW] &  =\operatorname{Tr}[C^{\prime}\rho],\\
\operatorname{Tr}[A_{i}W] &  =\operatorname{Tr}[A_{i}^{\prime}\rho],
\end{align}
and thus the desired equality in \eqref{eq:reduction-to-states_quantum}\ holds.

\section{Estimating the Hessian in free energy minimization}

\label{app:Hessian_est}

Let us first recall the expression for the $(i, j)$-th element of the Hessian of the objective function $f(\mu)$ defined in~\eqref{eq:obj_fun}:
\begin{equation}
  \frac{\partial^{2}}{\partial\mu_{i}\partial\mu_{j}}f(\mu)
    = \frac{1}{T}\left(  
    \left\langle Q_{i}\right\rangle _{\rho_{T}(\mu)}
    \left\langle Q_{j}\right\rangle _{\rho_{T}(\mu)} -\frac{1}{2}\left\langle 
    \left\{  \Phi_{\mu}(Q_{i}),Q_{j}\right\}
    \right\rangle _{\rho_{T}(\mu)}\right).
  \label{eq:hessian-app}
\end
{equation}
The factor $\sfrac{1}{T}$ is an overall scalar that does not affect the structure of the quantum algorithms that estimate the quantity in~\eqref{eq:hessian-app}. For simplicity, we therefore describe the algorithms for estimating the quantity inside the parentheses of the right-hand side of the above equation, and multiply the final output by $\sfrac{1}{T}$ in classical post‑processing.

Estimating the first term of the right-hand side of the above equation is relatively straightforward (see~\cite[Algorithm 2]{patel2024quantumboltzmannmachinelearning} for similar quantity estimation). So, in what follows, we present an algorithm for estimating the second term in greater detail.

Consider the following:
\begin{align}
    -\frac{1}{2}\left\langle \left\{  \Phi_{\mu}(Q_{i}), Q_j \right\}  \right\rangle_{\rho_T(\mu)}
    & = - \frac{1}{2}\operatorname{Tr}[\left\{
    \Phi_{\mu}(Q_{i}), Q_j \right\}  \rho_T(\mu)]\\
    & = \int_{\mathbb{R}} dt\ p(t)\ \left( -\frac{1}{2} \operatorname{Tr}\!\left[ \left\{e^{-i\left(  H-\mu\cdot Q\right)  t/T} Q_i e^{i\left(  H-\mu\cdot Q\right)  t/T} , Q_j \right\} \rho_T(\mu)\right] \right) \\
    & = \sum_{\vec{\ell}} \sum_{\vec{k}} a_{i, \vec{\ell}} \ a_{j, \vec{k}} \int_{\mathbb{R}} dt\ p(t)\ \left( -\frac{1}{2} \operatorname{Tr}\!\left[ \left\{e^{-i\left(  H-\mu\cdot Q\right)  t/T} \sigma_{\vec{\ell}} \ e^{i\left(  H-\mu\cdot Q\right)  t/T} , \sigma_{\vec{k}} \right\} \rho_T(\mu)\right] \right),\label{eq:first_hess_simplify}
\end{align}
where in the last equality we used the fact that $Q_i$ and $Q_j$ can be written as linear combinations of tensor products of Pauli matrices as assumed in~\eqref{eq:Q_pauli}, that is, $Q_i = \sum_{\vec{\ell}} a_{i, \vec{\ell}} \, \sigma_{\vec{\ell}}$ and $Q_j = \sum_{\vec{k}} a_{j, \vec{k}} \, \sigma_{\vec{k}}$.

We are now in a position to present an algorithm (Algorithm~\ref{algo:hessian_est}) for estimating the second term of~\eqref{eq:hessian-app} by using its equivalent form shown in~\eqref{eq:first_hess_simplify}. At the core of our algorithm lies the quantum circuit that estimates the expected value of the anticommutator of two operators presented in~\cite[Appendix B]{minervini2025evolvedquantumboltzmannmachines}. In particular,~\cite[Figure 5.a]{minervini2025evolvedquantumboltzmannmachines} provides the exact circuit required in our case, as it estimates the quantity $-\frac{1}{2} \left\langle \left\{ U, H\right\} \right\rangle_{\rho}$, where $H$ is Hermitian and $U$ is Hermitian and unitary. In our case, we choose $\rho = \rho_T(\mu)$, $U = e^{-i\left(  H-\mu\cdot Q\right)  t/T} \sigma_{\vec{\ell}} \ e^{i\left(  H-\mu\cdot Q\right)  t/T}$, and $H = \sigma_{\vec{k}}$. We then make some further simplifications that follow because $\rho_T(\mu)$ commutes with $e^{-i\left(  H-\mu\cdot Q\right)t/T}$. Accordingly, the quantum circuit that estimates the integrand of~\eqref{eq:first_hess_simplify} is depicted in Figure~\ref{fig:KM}. The algorithm (Algorithm~\ref{algo:hessian_est}) involves running this circuit $N$ times, where $N$ is determined by the desired accuracy and error probability. During each run, the time $t$ for the Hamiltonian evolution is sampled at random with probability $p(t)$ (defined in~\eqref{eq:high-peak-tent-density}), while the indexes $\vec{\ell}$ and $\vec{k}$ are sampled with probability $\sfrac{a_{i,\vec{\ell}}}{\left \Vert a_i \right \Vert_1}$ and $\sfrac{a_{j,\vec{k}}}{\left \Vert a_j \right \Vert_1}$ respectively. The final estimation of the second term of~\eqref{eq:hessian-app} is obtained by averaging the outputs of the $N$ runs.

\begin{algorithm}[H]
\caption{\texorpdfstring{$\mathtt{estimate\_anticommutator\_hessian}(\mu, (a_{i,\vec{\nu}})_{\vec{\nu}}, (a_{j,\vec{\nu}})_{\vec{\nu}}, \varepsilon, \delta)$}{estimate first term}}
\label{algo:hessian_est}
\begin{algorithmic}[1]
\STATE \textbf{Input:} chemical potential vector $\mu$, non-negative coefficients $(a_{i,\vec{\nu}})_{\vec{\nu}}$ and $(a_{j,\vec{\nu}})_{\vec{\nu}}$ for the Pauli strings $(\sigma_{\vec{\nu}})_{\vec{\nu}}$, accuracy $\varepsilon > 0$, error probability $\delta \in (0,1)$
\STATE $N \leftarrow \lceil\sfrac{2 \left \Vert a_i \right \Vert_1^2 \left \Vert a_j \right \Vert_1^2\ln(\sfrac{2}{\delta})}{\varepsilon^2}\rceil$
\FOR{$n = 0$ to $N-1$}
\STATE Initialize the control register to $|1\rangle\!\langle 1 |$
\STATE Prepare the system register in the state $\rho_T(\mu)$
\STATE Sample $\vec{\ell}$, $\vec{k}$, and $t$ with probability $\sfrac{a_{i,\vec{\ell}}}{\left \Vert a_i \right \Vert_1}$, $\sfrac{a_{j,\vec{k}}}{\left \Vert a_j \right \Vert_1}$, and $p(t)$ (defined in~\eqref{eq:high-peak-tent-density}), respectively
\STATE Apply the Hadamard gate to the control register
\STATE Apply the following unitaries to the control and system registers:
\STATE \hspace{0.6cm} \textbullet~Controlled-$\sigma_{\vec{k}}$: $\sigma_{\vec{k}}$ is a local unitary acting on the system register, controlled by the control register
\STATE \hspace{0.6cm} \textbullet~$e^{i\left(  H-\mu\cdot Q\right)  t/T}$: Hamiltonian simulation for time $t$ on the system register
\STATE Apply the Hadamard gate to the control register
\STATE Measure the control register in the computational basis and store the measurement outcome~$\lambda_n$
\STATE Measure the system register in the eigenbasis of $\sigma_{\vec{\ell}}$ and store the measurement outcome $\gamma_n$
\STATE $Y_{n} \leftarrow (-1)^{\lambda_n + \gamma_n}$
\ENDFOR

\STATE \textbf{return} $\overline{Y} \leftarrow  \left \Vert a_i \right \Vert_1 \left \Vert a_j \right \Vert_1 \frac{1}{N}\sum_{n=0}^{N-1}Y_{n}$
\end{algorithmic}
\end{algorithm}

\section{Sample Complexity of the AG19 Quantum MMW Algorithm  for Solving SDPs}

\label{app:samp-comp-of-prior-work}

In this section, we provide a more detailed explanation of how we arrived at the sample complexity expression in~\eqref{eq:sample-comp-prior-method}, as it is not directly evident from the expression in \cite[Theorem~8]{vanApeldoorn2019arXiv}.

The algorithm of~\cite{Apeldoorn2019} is an iterative method consisting of $O\!\left(\gamma^2 \ln d\right)$ iterations, where $\gamma = \Theta(\sfrac{rR}{\varepsilon})$. Each iteration requires $ O\!\left(\gamma^2\right)$ thermal-state samples of the form in~\eqref{eq:grand-canon-thermal-state} to be used for two subroutines: Two-Phase Quantum Search and Two-Phase Quantum Minimum Find. Consequently, the total sample complexity of the AG19 algorithm is 
\begin{equation}
    O\!\left(\gamma^4 \ln d\right) = O\!\left (\frac{r^4R^4}{\varepsilon^4} \ln d\right).
\end{equation}
Finally, note that all the results in~\cite{Apeldoorn2019} are stated for a constant success probability of $2/3$. To boost the success probability to $1-\delta$ for arbitrary $\delta \in (0, 1)$, it suffices to repeat the procedure $O(\ln (\sfrac{1}{\delta}))$ times.
This implies that the final sample complexity of this algorithm with success probability $1-\delta$ is then
\begin{equation}
    O\!\left(\frac{r^4R^4 \ln\! \left(\sfrac{1}{\delta}\right) \ln d}{\varepsilon^4} \right),
\end{equation} 
which is the expression in~\eqref{eq:sample-comp-prior-method}.

\end{document}